\pdfoutput=1
\let\arxiv=\relax

\documentclass[a4paper,UKenglish]{lipics-v2018}
\nolinenumbers

\usepackage{microtype}

\title{A Subquadratic Algorithm for 3XOR}

\author{Martin Dietzfelbinger}{Technische Universität Ilmenau, Germany}{martin.dietzfelbinger@tu-ilmenau.de}{https://orcid.org/0000-0001-5484-3474}{}%
\author{Philipp Schlag}{Technische Universität Ilmenau, Germany}{philipp.schlag@tu-ilmenau.de}{https://orcid.org/0000-0001-5052-9330}{}
\author{Stefan Walzer}{Technische Universität Ilmenau, Germany}{stefan.walzer@tu-ilmenau.de}{https://orcid.org/0000-0002-6477-0106}{}

\authorrunning{M. Dietzfelbinger, P. Schlag and S. Walzer}

\Copyright{Martin Dietzfelbinger, Philipp Schlag, Stefan Walzer}%

\subjclass{ Theory of computation $\rightarrow$ Design and analysis of algorithms}%


\category{}

\relatedversion{}

\supplement{}

\funding{}

\EventEditors{John Q. Open and Joan R. Access}
\EventNoEds{2}
\EventLongTitle{42nd Conference on Very Important Topics (CVIT 2016)}
\EventShortTitle{CVIT 2016}
\EventAcronym{CVIT}
\EventYear{2016}
\EventDate{December 24--27, 2016}
\EventLocation{Little Whinging, United Kingdom}
\EventLogo{}
\SeriesVolume{42}
\ArticleNo{}

\usepackage[capitalize,noabbrev]{cleveref}
\usepackage{thm-restate}
\usepackage{xspace}

\makeatletter
\g@addto@macro\bfseries{\boldmath}
\makeatother
\newcommand{\red}[1]{#1}

\theoremstyle{plain}
\newtheorem{mylemma}[theorem]{Lemma}
\crefname{mylemma}{Lemma}{Lemmata}
\newtheorem{mycorollary}[theorem]{Corollary}
\crefname{mycorollary}{Corollary}{Corollaries}
\theoremstyle{plain}
\newtheorem{problem}[theorem]{Problem}
\crefname{problem}{Problem}{Problems}
\crefname{algorithm}{Algorithm}{Algorithms}
\crefname{algocf}{Algorithm}{Algorithms}
\theoremstyle{remark}
\newtheorem{claim}[theorem]{Claim}
\crefname{claim}{Claim}{Claims}

\newcommand{\bucket}[2]{#1_{#2}}
\newcommand{\wpa}[2]{#1^{*}_{#2}}
\newcommand{\wpax}[3]{#1^{*, #3}_{#2}}
\newcommand{\good}[1]{#1^{\mathrm{g}}}
\newcommand{\bad}[1]{#1^{\mathrm{b}}}

\newcommand{\wordRAM}{word~RAM}
\newcommand{\WordRAM}{Word~RAM}

\newcommand{\threeSUM}{3SUM\xspace}
\newcommand{\intthreeSUM}{int3SUM\xspace}
\newcommand{\threeXOR}{3XOR\xspace}

\newcommand{\lcp}{\mathrm{lcp}}
\newcommand{\bigO}{O}
\newcommand{\hashFamily}{\mathcal{H}}

\newcommand{\SetDisjointness}{\textsf{SetDisjointness}}
\newcommand{\SetIntersection}{\textsf{SetIntersection}}
\newcommand{\offSetDisjointness}{offline \SetDisjointness}
\newcommand{\offSetIntersection}{offline \SetIntersection}
\newcommand{\OffSetDisjointness}{Offline \SetDisjointness}
\newcommand{\OffSetIntersection}{Offline \SetIntersection}

\newcommand{\floor}[1]{\left\lfloor #1\right\rfloor}
\newcommand{\ceil}[1]{\left\lceil #1\right\rceil}
\newcommand{\abs}[1]{\lvert #1\rvert}
\newcommand{\card}[1]{\lvert #1\rvert}

\newcommand{\tree}[1]{T_{#1}}

\newcommand{\logrange}{\mu}
\newcommand{\logdom}{\ell}

\newcommand{\suchthat}{s.t.}
\newcommand{\wlogMy}{w.l.o.g.}
\newcommand{\Wlog}{W.l.o.g.}
\newcommand{\eg}{e.g.}

\newcommand{\ie}{i.e.}

\newcommand{\ignore}[1]{}

\newcommand{\condSet}[2]{\{\begin{gathered}\,#1\,\end{gathered}\mid\begin{gathered}\,#2\,\end{gathered}\}}
\newcommand{\prob}[2]{\mathrm{\mathbf{Pr}}_{#1}[\begin{gathered}\,#2\,\end{gathered}]}
\newcommand{\condProb}[3]{\mathrm{\mathbf{Pr}}_{#1}[\begin{gathered}\,#2\,\end{gathered}\mid\begin{gathered}\,#3\,\end{gathered}]}
\newcommand{\E}[2]{\mathrm{\mathbf{E}}_{#1}[\,#2\,]}

\usepackage{amsfonts,amsopn}
\DeclareMathOperator{\AND}{\textsc{and}}
\DeclareMathOperator{\OR}{\textsc{or}}
\DeclareMathOperator{\NOT}{\textsc{not}}
\DeclareMathOperator{\XOR}{\textsc{xor}}
\DeclareMathOperator{\xor}{\oplus}

\DeclareMathOperator{\PARITY}{\textsc{Parity}}

\usepackage[disable,colorinlistoftodos,prependcaption,textsize=tiny]{todonotes}

\usepackage{array}
\newcolumntype{C}[1]{>{\centering\let\newline\\\arraybackslash\hspace{0pt}}m{#1}}

\usepackage{arydshln} 

\usepackage{tikz}
\usetikzlibrary{matrix}
\usetikzlibrary{calc,decorations.pathreplacing}
\newcommand{\tikzmark}[1]{\tikz[overlay,remember picture] \node (#1) {};}
\newcommand*{\AddNote}[4]{%
    \begin{tikzpicture}[overlay, remember picture]
        \draw [decoration={brace,amplitude=0.5em},decorate,ultra thick]
            ($(#3)!(#1.north)!($(#3)-(0,1)$)$) --  
            ($(#3)!(#2.south)!($(#3)-(0,1)$)$)
                node [align=center, text width=2.5cm, pos=0.5, anchor=west] {#4};
    \end{tikzpicture}
}%

\usepackage[utf8]{inputenc}
\usepackage{newunicodechar}
\usepackage{silence}
\usepackage[vlined,linesnumbered]{algorithm2e}
\DontPrintSemicolon
\SetKwProg{algo}{Algorithm}{:}{}
\SetKwProg{subroutine}{subroutine}{:}{}
\SetKwProg{reduction}{Reduction}{:}{}

\SetCommentSty{commentFont}

\newunicodechar{⁻}{^-}
\newunicodechar{⊕}{\oplus}
\newunicodechar{⁺}{^+}
\newunicodechar{₋}{_-}
\newunicodechar{₊}{_+}
\newunicodechar{ℓ}{\ell}
\newunicodechar{•}{\item}
\newunicodechar{…}{\dots}
\newunicodechar{≔}{\coloneqq}
\newunicodechar{≤}{\leq}
\newunicodechar{≥}{\geq}
\newunicodechar{≰}{\nleq}
\newunicodechar{≱}{\ngeq}
\newunicodechar{ϕ}{\varphi}
\newunicodechar{≠}{\neq}
\newunicodechar{¬}{\neg}
\newunicodechar{≡}{\equiv}
\newunicodechar{ₙ}{_n}
\newunicodechar{₀}{_0}
\newunicodechar{₁}{_1}
\newunicodechar{₂}{_2}
\newunicodechar{₃}{_3}
\newunicodechar{₄}{_4}
\newunicodechar{₅}{_5}
\newunicodechar{₆}{_6}
\newunicodechar{₇}{_7}
\newunicodechar{₈}{_8}
\newunicodechar{₉}{_9}
\newunicodechar{⁰}{^0}
\newunicodechar{¹}{^1}
\newunicodechar{²}{^2}
\newunicodechar{³}{^3}
\newunicodechar{⁴}{^4}
\newunicodechar{⁵}{^5}
\newunicodechar{⁶}{^6}
\newunicodechar{⁷}{^7}
\newunicodechar{⁸}{^8}
\newunicodechar{⁹}{^9}
\newunicodechar{∈}{\in}
\newunicodechar{∉}{\notin}
\newunicodechar{⊂}{\subset}
\newunicodechar{⊃}{\supset}
\newunicodechar{⊆}{\subseteq}
\newunicodechar{⊇}{\supseteq}
\newunicodechar{⊄}{\nsubset}
\newunicodechar{⊅}{\nsupset}
\newunicodechar{⊈}{\nsubseteq}
\newunicodechar{⊉}{\nsupseteq}
\newunicodechar{∪}{\cup}
\newunicodechar{∩}{\cap}
\newunicodechar{∀}{\forall}
\newunicodechar{∃}{\exists}
\newunicodechar{∄}{\nexists}
\newunicodechar{∨}{\vee}
\newunicodechar{∧}{\wedge}
\newunicodechar{ℝ}{\mathbb{R}}
\newunicodechar{ℕ}{\mathbb{N}}
\newunicodechar{ℤ}{\mathbb{Z}}
\newunicodechar{·}{\cdot}
\newunicodechar{∘}{\circ}
\newunicodechar{×}{\times}
\newunicodechar{↑}{\uparrow}
\newunicodechar{↓}{\downarrow}
\newunicodechar{→}{\rightarrow}
\newunicodechar{←}{\leftarrow}
\newunicodechar{⇒}{\Rightarrow}
\newunicodechar{⇐}{\Leftarrow}
\newunicodechar{↔}{\leftrightarrow}
\newunicodechar{⇔}{\Leftrightarrow}
\newunicodechar{↦}{\mapsto}
\newunicodechar{∅}{\emptyset}
\newunicodechar{∞}{\infty}
\newunicodechar{≅}{\cong}
\newunicodechar{≈}{\approx}

\newunicodechar{α}{\alpha}
\newunicodechar{β}{\beta}
\newunicodechar{γ}{\gamma}
\newunicodechar{Γ}{\Gamma}
\newunicodechar{δ}{\delta}
\newunicodechar{Δ}{\Delta}
\newunicodechar{ε}{\varepsilon}
\newunicodechar{ζ}{\zeta}
\newunicodechar{η}{\eta}
\newunicodechar{θ}{\theta}
\newunicodechar{Θ}{\Theta}
\newunicodechar{ι}{\iota}
\newunicodechar{κ}{\kappa}
\newunicodechar{λ}{\lambda}
\newunicodechar{Λ}{\Lambda}
\newunicodechar{μ}{\mu}
\newunicodechar{ν}{\nu}
\newunicodechar{ξ}{\xi}
\newunicodechar{Ξ}{\Xi}
\newunicodechar{π}{\pi}
\newunicodechar{Π}{\Pi}
\newunicodechar{ρ}{\rho}
\newunicodechar{σ}{\sigma}
\newunicodechar{Σ}{\Sigma}
\newunicodechar{τ}{\tau}
\newunicodechar{υ}{\upsilon}
\newunicodechar{ϒ}{\Upsilon}
\newunicodechar{ϕ}{\phi}
\newunicodechar{Φ}{\Phi}
\newunicodechar{χ}{\chi}
\newunicodechar{ψ}{\psi}
\newunicodechar{Ψ}{\Psi}
\newunicodechar{ω}{\omega}
\newunicodechar{Ω}{\Omega}

\ifdefined\arxiv
  \hideLIPIcs
  \makeatletter
  \let\@evenfoot\@empty
  \let\@oddfoot\@empty
  \def\@evenhead{\large\sffamily\bfseries
          \llap{\hbox to0.5\oddsidemargin{\thepage\hss}}\leftmark\hfil}%
  \def\@oddhead{\large\sffamily\bfseries\rightmark\hfil
         \rlap{\hbox to0.5\oddsidemargin{\hss\thepage}}}%
  \makeatother
\fi
\begin{document}

\maketitle
\begin{abstract}
  Given a set $X$ of $n$ binary words of equal length $w$, the \threeXOR{} problem 
	asks for three elements $a, b, c\in X$ such that $a\xor b=c$, where $\xor$ denotes the bitwise XOR operation.
  The problem can be easily solved on a word RAM with word length $w$ in time $\bigO(n^2 \log{n})$. 
	Using Han's fast integer sorting algorithm (2002/2004) this can be reduced to $\bigO(n^2 \log{\log{n}})$.
  With randomization or a sophisticated deterministic dictionary construction, 
	creating a hash table for $X$ with constant lookup time 
	leads to an algorithm with (expected) running time $\bigO(n^2)$.
  At present, seemingly no faster algorithms are known.
  
  We present a surprisingly simple deterministic, quadratic time algorithm for \threeXOR.
  Its core is a version of the Patricia trie for $X$, which makes it possible to traverse the set $a\xor X$ 
	in ascending order for arbitrary $a\in \{0, 1\}^{w}$ in linear time.
  Furthermore, we describe a randomized algorithm for \threeXOR{} with expected running time 
	$\bigO(n^2\cdot\min\{\frac{\log^3{w}}{w}, \frac{(\log\log{n})^2}{\log^2 n}\})$.
  The algorithm transfers techniques to our setting that were used by Baran, Demaine, and P{\u{a}}tra{\c{s}}cu (2005/2008) 
	for solving the related \intthreeSUM{} problem (the same problem with integer addition in
	place of binary XOR) in expected time $o(n^2)$. 
	As suggested by Jafargholi and Viola (2016), linear hash functions are employed.
  
  The latter authors also showed that assuming \threeXOR{} needs expected running time $n^{2-o(1)}$
	one can prove conditional lower bounds for triangle enumeration just as with \threeSUM.
  We demonstrate that \threeXOR{} can be reduced to other problems as well, treating 	
  the examples \offSetDisjointness{} and \offSetIntersection, which were 
	studied for \threeSUM{} by Kopelowitz, Pettie, and Porat (2016).
\end{abstract}

\section{Introduction}
\label{sec:intro}

The \threeXOR{} problem~\cite{JV16} is the following: Given a set $X$
of $n$ binary strings of equal length $w$, are there 
elements $a, b, c\in X$ such that  $a\xor b = c$,
where $\xor$ is bitwise XOR?
We work with the word RAM~\cite{FW93} model with word length $w=\Omega(\log n)$,
and we assume as usual that one input string fits into one word. 
Then, using sorting, the problem can easily be solved in time $\bigO(n^2\log{n})$.
Using Han's fast integer sorting algorithm~\cite{Han04} the time can be reduced to
$\bigO(n^2\log\log{n})$.
In order to achieve quadratic running time, 
one could utilize a randomized dictionary for $X$ with expected linear 
construction time and constant lookup time (like in~\cite{FKS84}) 
or (weakly non-uniform, quite complicated) deterministic static dictionaries with
construction time $\bigO(n\log n)$ and constant lookup time as provided in~\cite{HMP01}. 
Once such a dictionary is available, one just has to check whether $a \xor b \in X$, for all pairs $a,b \in X$.
No subquadratic algorithms seem to be known.

It is natural to compare the situation with that for the \threeSUM{} problem, 
which is as follows:%
\footnote{There are many different, but equivalent versions of \threeSUM{} and \threeXOR{},
differing in the way the input elements are grouped. Often one sees the demand that the 
three elements $a$, $b$, and $c$ with  $a\xor b=c$ or $a+b=c$, resp., come from different sets.}
Given a set $X$ of $n$ real numbers, are there 
$a, b, c\in X$ such that $a + b = c$? 
There is a very simple quadratic time algorithm for this problem (see Section~\ref{sec:quadraticDeterministic} below). 
After a randomized subquadratic algorithm was suggested by Gr{\o}nlund J{\o}rgensen and Pettie~\cite{JP14},
improvements ensued~\cite{F17,GS15}, and recently
Chan~\cite{Ch18} gave the fastest deterministic algorithm known,
with a running time of $n^2(\log\log n)^{\bigO(1)}/\log^2 n$.
The restricted version where the input consists of integers whose bit length does not
exceed the word length $w$ is called \intthreeSUM{}. 
The currently best randomized algorithm for \intthreeSUM{} was given by 
Baran, Demaine, and P{\v{a}}tra{\c{s}}cu~\cite{BDP08,BDP05};
it runs in expected time $\bigO(n^{2}\cdot\min\{\frac{\log^{2}{w}}{w}, 
\frac{(\log{\log{n}})^2}{\log^2{n}}\})$ for $w=\bigO(n\log{n})$. 
The \threeSUM{} problem has received a lot of attention in recent years, 
because it can be used as a basis for conditional lower time bounds
for problems \eg{} from computational geometry 
and data structures~\cite{GO95,KPP16,P10}.
Because of this property, \threeSUM{} is in the center of attention 
of papers dealing with low-level complexity. 
Chan and Lewenstein~\cite{CL15} give upper bounds for inputs with a certain structure. 
Kane, Lovett, and Moran~\cite{KLM17} prove near-optimal upper bounds for linear decision trees. 
Wang~\cite{W14} considers randomized algorithms for subset sum, trying to minimize the space, and
Lincoln et al.~\cite{LWWW16} investigate time-space tradeoffs in deterministic algorithms for $k$-SUM.

In contrast, \threeXOR{} received relatively little attention, before 
Jafargholi and Viola~\cite{JV16} studied \threeXOR{} and
described techniques for reducing this problem to triangle enumeration. 
In this way they obtained conditional lower bounds in a way similar to 
the conditional lower bounds based on int\threeSUM{}. 

The main results of this paper are the following:
We present a surprisingly simple deterministic algorithm for \threeXOR{} that runs in time $\bigO(n^2)$.
When $X$ is given in sorted order,
it constructs in linear time a version of the Patricia trie~\cite{Mor68} for $X$, using only word operations
and not looking at single bits at all. This tree then makes it possible to 
traverse the set $a\xor X$ in ascending order in linear time, for arbitrary $a\in\{0,1\}^w$.
This is sufficient for achieving running time $\bigO(n^2)$.
The second result is a randomized algorithm for \threeXOR{} 
that runs in time $\bigO(n^{2}\cdot\min\{\frac{\log^{3}{w}}{w}, 
\frac{(\log{\log{n}})^2}{\log^2{n}}\})$ for $w=\bigO(n\log{n})$, 
which is almost the same bound as that of~\cite{BDP08} for \intthreeSUM{}.
Finding a deterministic algorithm for \threeXOR{} with subquadratic running time remains an open problem.
Finally, we reduce \threeXOR{} to \offSetDisjointness{} and \offSetIntersection, establishing conditional lower bounds (as in~\cite{KPP16} conditioned on the \intthreeSUM{} conjecture).

Unfortunately, no (non-trivial) relation between the required (expected) time for \threeSUM{} and \threeXOR{} is known.
In particular, we cannot exclude the case that one of these problems can be solved 
in (expected) time $\bigO(n^{2-\varepsilon})$ for some constant $\varepsilon>0$
whereas the other one requires (expected) time $n^{2-o(1)}$.
Actually, this possibility is the background of some conditional statements 
on the cost of listing triangles in graphs in~\cite[Cor. 2]{JV16}.
However, due to the similarity of \threeXOR{} to \threeSUM{},
the question arises
whether the recent results on \threeSUM{} can be transferred to \threeXOR{}.

In \Cref{sec:preliminaries}, we review the \wordRAM{} 
model and examine $1$-universal classes of linear hash functions.
In particular, we determine the evaluation cost of such hash functions 
and we restate a hashing lemma \cite{BDP08} on the expected number of elements in ``overfull'' buckets.
Furthermore, we state how fast one can solve the set intersection problem on word-packed arrays
(with details given in the appendix).
In \Cref{sec:quadraticDeterministic}, we construct a special enhanced binary search tree $\tree{X}$ to represent a set $X$
of binary strings of fixed length. 
This representation makes it possible to traverse the set $a\xor X$ 
in ascending order for any $a\in \{0, 1\}^{w}$ in linear time, 
which leads to a simple deterministic algorithm for \threeXOR{} that runs in time $\bigO(n^2)$.
Then, we turn to randomized algorithms and 
show how to solve \threeXOR{} in subquadratic expected time in \Cref{sec:subquadraticRandomized}:
$\bigO(n^{2}\cdot\min\{\frac{\log^{3}{w}}{w}, 
\frac{(\log{\log{n}})^2}{\log^2{n}}\})$ for $w=\bigO(n\log{n})$, 
and $\bigO(n\log^{2}{n})$ for $n\log n\leq w=\bigO(2^{n\log n})$.
Our approach uses the ideas of the subquadratic expected time algorithm for \intthreeSUM{} 
presented in \cite{BDP08}, \ie, computing buckets and fingerprints, word packing, 
exploiting word-level parallelism, and using lookup tables.
Altogether, we get the same expected running time for $w=\bigO(\log^2{n})$ 
and a word-length-dependent upper bound on the expected running time for $w=\omega(\log^2{n})$ that
is worse by a $\log w$ factor in comparison to the \intthreeSUM{} setting.
Based on these results and the similarity of \threeXOR{} to \threeSUM{}, 
it seems natural to conjecture that \threeXOR{} requires expected time $n^{2-o(1)}$, too,
and so \threeXOR{} is a candidate for reductions to other computational problems just as \threeSUM.
In \Cref{sec:condLowerBounds}, we describe 
how to reduce \threeXOR{} to \offSetDisjointness{} and \offSetIntersection{}, 
transferring the results of \cite{KPP16} from \threeSUM{} to \threeXOR.

Recently, Bouillaguet et al.~\cite{BDF18} studied algorithms ``for the 3XOR problem''.
This is related to our setting, but not identical. 
These authors study a variant of the ``generalized birthday problem'', well known in cryptography
as a problem to which some attacks on cryptosystems can be reduced, see~\cite{BDF18}. 
Translated into our notation, their question is:
Given a \emph{random} set $X\subseteq \{0,1\}^w$ of size $3\cdot2^{w/3}$, find, if possible,
three different strings $a,b,c\in X$ such that $a\xor b=c$.
Adapting the algorithm from~\cite{BDP08}, these authors achieve a running time of 
$\bigO(2^{2w/3}(\log^2 w)/w^2)$, which corresponds to the running time 
of our algorithm for $n=3 \cdot 2^{w/3}$.
The difference to our situation is that their input is random.
This means that the issue of 1-universal families of linear hash functions disappears
(a projection of the elements in $X$ on some bit positions does the job) and that 
complications from weak randomness are absent
(e.g., one can use projection into relatively small buckets and use 
Chernoff bounds to prove that the load is very even with high probability).
This means that the algorithm described in~\cite{BDF18} 
does not solve our version of the \threeXOR{} problem.


\section{Preliminaries}
\label{sec:preliminaries}

\subsection{The \WordRAM{} Model}
\label{subsec:wordRAM}

As is common in the context of fast algorithms for the \intthreeSUM{} problem~\cite{BDP08}, 
we base our discussion on the \emph{\wordRAM{} model}~\cite{FW93}. 
This is characterized by a word length $w$.
Each memory cell can store $w$ bits, 
interpreted as a bit string or an integer or a packed sequence of subwords, as is convenient.
The word length $w$ is assumed to be at least $\log n$ and at least
the bit length of a component of the input.
It is assumed that the operations of the \emph{multiplicative instruction set}, \ie{}, 
arithmetic operations (addition, subtraction, multiplication), 
word operations (left shift, right shift), 
bitwise Boolean operations ($\AND$, $\OR$, $\NOT$, $\XOR$), 
and random memory accesses can be executed in constant time.
We will write $\xor$ to denote the bitwise $\XOR$ operation.
A \emph{randomized} word RAM also provides an operation that 
in constant time generates a
uniformly random value in $\{0,1,\dots,v-1\}$ for any given $v\leq 2^w$.

\subsection{Linear Hash Functions}
\label{subsec:LinearUniversalHashing}
We consider hash functions $h \colon U \to M$,
where the domain (``universe'') $U$ is $\{0,1\}^\logdom$ and the range $M$ 
is $\{0,1\}^\logrange$ with $\logrange\le \logdom$. 
Both universe and range are vector spaces over $\mathbb{Z}_2$.
In~\cite{BDP08} and in successor papers on \intthreeSUM{} ``almost linear'' hash functions
based on integer multiplication and truncation were used, as can be found in~\cite{D96}.
As noted in ~\cite{JV16}, in the \threeXOR{} setting the situation is much simpler. 
We may use $\hashFamily^{\text{lin}}_{\logdom,\logrange}$, 
the set of all $\mathbb{Z}_2$-linear functions from $U$ to $M$.
A function $h_A$ from this family is described by a $\logrange\times \logdom$ matrix $A$,
and given by $h_A(x) = A \cdot x$, where $x = (x_0,\dots,x_{\logdom-1})^{\textsf{T}}\in U$ 
and $h_A(x)\in M$ are written as column vectors.
For all hash functions $h\in\hashFamily^{\text{lin}}_{\logdom,\logrange}$ and all $x,y\in U$ we have 
$h(x\xor y) = h(x) \xor h(y)$, by the very definition of \emph{linearity}.
Further, this family is $1$\emph{-universal},
indeed, we have $\prob{A\in\{0,1\}^{\logrange\times\logdom}}{h_A(x)=h_A(y)} = 
\prob{A\in\{0,1\}^{\logrange\times\logdom}}{h_A(x\xor y)=0} = 2^{-\logrange} = 1/\card{M}$, 
for all pairs $x,y$ of different keys in $U$.
We remark that the convolution class described in~\cite{MNT93}, 
a subfamily of $\hashFamily^{\text{lin}}_{\logdom,\logrange}$,
can be used as well, as it is also 1-universal, and needs only $\logdom+\logrange-1$ random bits.  
\ignore{%
A disadvantage of class $\hashFamily^{\text{lin}}_{\logdom,\logrange}$ 
is that sampling a hash function from
it at random requires $\logdom \logrange$ random bits. 
We mention that by employing convolution one gets a smaller class,
$\hashFamily^{\text{conv}}_{\logdom,\logrange}$, also consisting of linear functions.
A function $h_{a}$ from this class is determined by a vector $a\in\{0, 1\}^{\logdom+\logrange-1}$.
From $a$ we form the matrix $A_a=(a_{i+j})_{0\le i < \logrange, 0\le j < \logdom}$, and 
$h_a$ is simply $h_{A_a}$, which means $h_a(x) = A_a \cdot x$ for $x\in U$. 
In~\cite{MNT93} it was shown that 
$\hashFamily^{\text{conv}}_{\logdom,\logrange}$ 
is 1-universal, \ie, $\prob{a\in\{0,1\}^{\logdom+\logrange-1}}{h_a(x)=h_a(y)} = 2^{-\logrange}$.}

The universe we consider here is $\{0, 1\}^{w}$.
The time for evaluating a hash function $h\in\hashFamily^{\text{lin}}_{w,\logrange}$ 
on one or on several inputs depends on the instruction set and on the way $h=h_A$ is stored.
In contrast to the \intthreeSUM{} setting \cite{BDP08}, we are not able to calculate hash values in constant time.

\begin{mylemma}\label{lem:evalTimeHashing} For \emph{$h\in\hashFamily^{\text{lin}}_{w,\logrange}$} 
and inputs from $\{0,1\}^w$ we have:\\
\emph{(a)} $h(x)$ can be calculated in time $\bigO(\logrange)$, if
$\PARITY$ of $w$-bit words is a constant time operation.\\
\emph{(b)} $h(x)$ can always be calculated in time $\bigO(\logrange + \log w)$.\\
\emph{(c)} $h(x_1),\dots,h(x_n)$ can be evaluated in time $\bigO(n\logrange + \log w)$.
\end{mylemma}
\begin{proof}(Sketch.) Assume $h=h_A$. For (a) we store the rows of $A$ as $w$-bit strings,
and obtain each bit of the hash value by a bitwise $\wedge$ operation followed by $\PARITY{}$. 
For (b) we assume the $w$ columns of $A$ are stored as $\logrange$-bit blocks, 
in $O(\logrange)$ words. An evaluation 
is effected by selecting the columns indicated by the 1-bits of $x$
and calculating the $\xor$ of these vectors in a word-parallel fashion.
In $\log{w}$ rounds, these vectors are added, halving the number of vectors in each round.
For (c), we first pack the columns selected for the $n$ input strings into $\bigO(n \logrange)$ 
words and then carry out the calculation indicated in (b),
but simultaneously for all $x_i$ and within as few words as possible.
This makes it possible to further exploit 
word-level parallelism, if $\logrange$ should be much smaller than $w$. 
\end{proof}
We shall use linear, 1-universal hashing for splitting the input set into buckets and
for replacing keys by fingerprints in \cref{sec:subquadraticRandomized}.

\begin{remark}
  \label{rem:wordLength}
  In the following, we will apply \cref{lem:evalTimeHashing}(c) to map $n$ binary strings of length $w$ 
	to hash values of length $\logrange=\bigO(\log n)$ in time $\bigO(n\log n + \log w)$.
  Since $\log{w}$ will dominate the running time only for huge word lengths, 
	we assume in the rest of the paper
	that $w=2^{\bigO(n\log n)}$ and that 
	all hash values can be calculated in time $\bigO(n\log n)$.
\end{remark}

\begin{remark}
  \label{rem:hashTable}
  When randomization is allowed, we will assume that we have constructed in expected $\bigO(n)$ time
	a standard hash table for input set $X$ with constant lookup time~\cite{FKS84}.
	(Arbitrary 1-universal classes can be used for this.)
\end{remark}

\subsection{A Hashing Lemma for 1-Universal Families}
\label{subsec:hashingLemma}

A hash family $\hashFamily$ of functions from $U$ to $M$ is called \emph{1-universal}
if $\prob{h\in\hashFamily}{h(x)=h(y)}\le 1/\card{M}$ for all $x,y\in U$, $x\neq y$.
We map a set $S\subseteq U$ with $\abs{S}=n$ into $M$ with $\abs{M}=m$ by a random element $h\in\hashFamily$.
In~\cite[Lemma 4]{BDP08} it was noted that for 1-universal families
the expected number of keys that collide with more than 
$3n/m$ other keys is bounded by $O(m)$. We state a slightly stronger version of that lemma.
(The strengthening is not essential for the application in the present paper.)
\begin{restatable}[slight strengthening of Lemma 4 in~\cite{BDP08}]{mylemma}{HashingLemma}
\label{lem:hashingLemma}
  Let $\hashFamily$ be a 1-universal class of hash functions
	from $U$ to $M$, with $m=\abs{M}$, and let $S\subseteq U$ with $\abs{S}=n$. Choose
  $h \in \hashFamily$ uniformly at random. For $i\in M$ define $B_i=\{y\in S \mid h(y)=i\}$.
  Then for $2\frac{n}{m} < t \leq n$ we have:
  \[\E{h\in\hashFamily}{\abs{\{x\in S \mid \abs{B_{h(x)}} \ge  t\}}} < \frac{n}{t-2\frac{n}{m}}\,.\]
\end{restatable}
(The bound in~\cite{BDP08} was about twice as large.
The proof is given in \Cref{subsec:app:hashingLemma}.)

In our algorithm, we will be interested in the number of elements in buckets with size 
at least three times the expectation.
Choosing $t=3\frac{n}{m}$ in \cref{lem:hashingLemma}, 
we conclude that the expected number of such elements is smaller than the number of buckets.

\begin{mycorollary}
  \label{cor:hashingLemma}
  In the setting of \cref{lem:hashingLemma} we have $\E{h\in\hashFamily}{\;\card{\condSet{x\in S}{\card{B_{h(x)}} \geq 3{n}/{m}}}\;}< m$.
\end{mycorollary}

\subsection{Set Intersection on Unsorted Word-Packed Arrays}
\label{subsec:setIntersectionWPA}

We consider the problem ``\emph{set intersection on unsorted word-packed arrays}'': 
Assume $k$ and $\ell$ are such that $k(\ell+\log{k}) \le w$,
and that two words $a$ and $b$ are given that both contain $k$ many $\ell$-bit strings:
$a$ contains $a_0,\dots,a_{k-1}$ and $b$ contains $b_0,\dots,b_{k-1}$. 
We wish to determine whether $\{a_0,\dots,a_{k-1}\}\cap\{b_0,\dots,b_{k-1}\}$ is empty or not
and find an element in the intersection if it is nonempty.

In~\cite[proof of Lemma 3]{BDP05} a similar problem is considered:
It is assumed that $a$ is sorted and $b$ is \emph{bitonic}, 
meaning that it is a cyclic rotation of a sequence that first grows and then falls. 
In this case one sorts the second sequence by a word-parallel version of bitonic merge
(time $\bigO(\log k)$), and then merges the two sequences into one sorted sequence
(again in time $\bigO(\log k)$). Identical elements now stand next to each other,
and it is not hard to identify them.
We can use a slightly slower modification of the approach of~\cite{BDP05}:
We sort both sequences by word-packed bitonic sort~\cite{AH97}, which takes time $\bigO(\log^2 k)$, 
and then proceed as before.%
\footnote{It is this slower version of packed intersection that causes our randomized \threeXOR{} algorithm
to be a little slower than the \intthreeSUM{} algorithm for $w=\Omega(\log^2{n})$.} We obtain the following result.

\begin{restatable}{mylemma}{SetIntersectionWPA}
  \label{lem:subproblemsWordPackedArrays}
  Assume $k(\ell+\log{k}) =\bigO(w)$, and   
  assume that two sequences of $\ell$-bit strings, each of length $k$, are given. 
  Then the $t$ entries that occur in both sequences 
  can be listed in time $\bigO(\log^2 k + t)$.
\end{restatable}
For completeness, we give a more detailed description in \Cref{subsec:subproblemsWordPackedArrays}.

\section{A Deterministic \threeXOR{} Algorithm in Quadratic Time}
\label{sec:quadraticDeterministic}

A well known deterministic algorithm for solving the \threeSUM{} problem in time $\bigO(n^{2})$ is reproduced in \cref{algo:3sumSimple}.
\begin{figure}[htpb]
    \begin{minipage}[b]{0.425\textwidth}
        \begin{algorithm}[H]
          \algo{\normalfont{3SUM}($X$)}{
            sort $X$ as $x₁ < \dots < x_n$\;
            \For{$a \in X$}{
              $(i,j) ← (1,1)$\;
              \While{$i ≤ n \AND j ≤ n$}{
                \uIf{$a + x_i < x_j$}{$i ← i +1$}
                \uElseIf{$a + x_i > x_j$}{$j ← j+1$}
                \lElse{\Return $(a,x_i,x_j)$}
              }
            }
            \Return \emph{no solution}\;
          }
          \caption{A simple quadratic 3SUM algorithm.}
          \label{algo:3sumSimple}
        \end{algorithm}
    \end{minipage}
    ~
    \begin{minipage}[b]{0.60\textwidth}
        \begin{algorithm}[H]
          \algo{\normalfont{3XOR}($X$)}{
            sort $X$ as $x₁ < \dots < x_n$\;
            $T_X ← \mathrm{makeTree}(X)$\;
            \For{$a \in X$}{
              $(i,j) ← (1,1)$\;
              $(y_i)_{1 ≤ i ≤ n} ← \mathrm{traverse}(T_X, a)$\;
              \While{$i ≤ n \AND j ≤ n$}{
                \uIf{$y_i < x_j$}{$i ← i +1$}
                \uElseIf{$y_i > x_j$}{$j ← j+1$}
                \lElse{\Return $(a,y_i ⊕ a,x_j)$}
              }
            }
            \Return \emph{no solution}\;
          }
          \caption{A quadratic 3XOR algorithm.}
          \label{algo:3xorSimple}
        \end{algorithm}
    \end{minipage}
\end{figure}
After sorting the input $X$ as $x_1<\dots<x_n$ in time $O(n\log n)$, 
we consider each $a\in X$ separately and look for triples of the form $a+b = c$. 
Such triples correspond to elements of the intersection of 
$a+X = \{a + x_1, \dots , a+x_n\}$ and $X$. 
Since $X$ is sorted, we can iterate over both $X$ and $a + X$ in ascending order 
and compute the intersection with an interleaved linear scan.

Unfortunately, the $\xor$-operation is not order preserving, \ie, 
$x < y$ does not imply $a \xor x < a \xor y$ for the lexicographic ordering on bitstrings---or, indeed, 
any total ordering on bitstrings.
We may sort $X$ and each set $a \xor X= \{a \xor x \mid x \in X\}$, for $a \in X$, 
separately to obtain an algorithm with running time $O(n^2\log n)$. 
Using fast deterministic integer sorting~\cite{Han04} reduces this to time $O(n^2\log\log n)$.
In order to achieve quadratic running time, 
one may utilize a randomized dictionary for $X$ with expected linear 
construction time and constant lookup time (like in~\cite{FKS84}) 
or (weakly non-uniform, rather complex) deterministic static dictionaries with
construction time $\bigO(n\log n)$ and constant lookup time as provided in~\cite{HMP01}. 
Once such a dictionary is available, one just has to check whether $a \xor b \in X$, for all $a,b \in X$.

Here we describe a rather simple deterministic algorithm with quadratic running time.
For this, we utilize a special binary search tree\footnote{The structure of the tree
is that of the Patricia trie~\cite{Mor68} for $X$.}
$T_X$ that allows, for arbitrary $a \in \{0,1\}^w$, 
to traverse the set $a ⊕ X = \{a ⊕ x \mid x \in X\}$ 
in lexicographically ascending order, in linear time. 
For $X ≠ ∅$, the tree $T_X$ is recursively defined as follows.
\begin{itemize}
        • If $X = \{x\}$, then $T_X$ is $\mathsf{LeafNode}(x)$, a tree consisting of a single leaf with label $x$.
        • If $|X| ≥ 2$, let $\lcp(X)$ denote the longest common prefix of the elements of $X$ 
				when viewed as bitstrings. That is, all elements of $X$ coincide on the first $k = \lvert\lcp(X)\rvert$ bits, 
				the elements of some nonempty set $X_0 \subsetneq X$ start with $\lcp(X)0$ 
				and the elements of $X_1 = X - X_0$ start with $\lcp(X)1$. 
				We define $T_X = \mathsf{InnerNode}(T_{X₀}, 0^k 1b, T_{X₁})$ for some $b\in\{0,1\}^{w-k-1}$, 
				meaning that $T_X$ consists of a root vertex with label $\ell = 0^k1 b$,
				a left subtree $T_{X₀}$ and a right subtree $T_{X₁}$.
				The choice of $b$ is irrelevant, but it is convenient to define the label more concretely 
				as $ℓ = (\max X_0)\xor(\min X_1)$.
\end{itemize}
Note that along paths of inner nodes down from the root the labels when regarded as integers are strictly decreasing.
We give an example in \cref{fig:treeRep} and  provide 
a $\bigO(n\log n)$ time construction of $T_X$ from $X$ in \cref{algo:makeTree}.

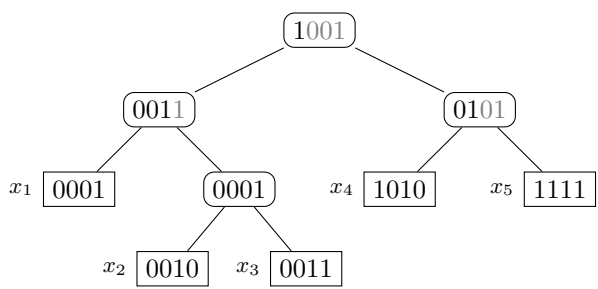
\begin{figure}[htbp]
  \centering
  \begin{tikzpicture}[
  level distance=3em,
  level 1/.style={sibling distance=12em},
  level 2/.style={sibling distance=6em}, 
  level 3/.style={sibling distance=5em},
  every node/.style = {shape=rectangle, draw, align=center},
  inner/.style = {rounded corners},
  leaf/.style = {}]
  \node[inner] {1\color{gray}001}
  child { node[inner] {001\color{gray} 1}
    child { node[leaf,label=left:\footnotesize$x₁$] {0001} }
    child { node[inner] {0001\color{gray} }
      child { node[leaf,label=left:\footnotesize$x₂$] {0010} }
      child { node[leaf,label=left:\footnotesize$x₃$] {0011} }
    }
  }
  child { node[inner] {01\color{gray} 01}
    child { node[leaf,label=left:\footnotesize$x₄$] {1010} }
    child { node[leaf,label=left:\footnotesize$x₅$] {1111} }
  };
  \end{tikzpicture}
  \caption{The tree $\tree{X}$ for $X=\{x₁ = 0001, x₂ = 0010, x₃ = 0011, x₄ = 1010, x₅ = 1111\}$.
      The first $1$-bit of the label of an inner node indicates the most significant bit that is not constant among the $x$-values managed by that subtree (the bits after the first $1$-bit are \textcolor{gray}{irrelevant}). According to the value of this bit, elements are found in the left or right subtree. Apart from the labels of the inner nodes,
			$\tree{X}$ is essentially the Patricia trie~\cite{Mor68} for $X$.
    }
  \label{fig:treeRep}
\end{figure}

In the context of $T_X = \mathsf{InnerNode}(T_{X₀}, ℓ = 0^k 1b, T_{X₁})$ as described above, the $(k{+}1)$st bit is the most significant bit where elements of $X$ differ. Crucially, this is also true for the set $a ⊕ X$ for any $a \in \{0,1\}^w$. Since the elements of $X$ are partitioned into $X₀$ and $X₁$ according to their $(k{+}1)$st bit, either all elements of $a ⊕ X₀$ are less than all elements of $a ⊕ X₁$, or vice versa, depending on whether the $(k{+}1)$st bit of $a$ is $0$ or $1$.
Using that the $(k{+}1)$st bit of $a$ is $1$ iff $a ⊕ ℓ < a$, this suggests a simple recursive algorithm to produce $a ⊕ X$ in sorted order, given as \cref{algo:traverse}.

\begin{figure}
    \begin{minipage}{0.5\textwidth}
        \begin{algorithm}[H]
          \SetKw{Yield}{yield}
          \algo{\normalfont{traverse}($T$, $a$)}{
            \uIf{$T = \mathsf{LeafNode}(x)$}{
              \Yield $a ⊕ x$
            }\Else{
              let $T = \mathsf{InnerNode}(T₀,ℓ,T₁)$\;
              \uIf{$a ⊕ ℓ > a$}{
                $\mathrm{traverse}(T₀,a); \mathrm{traverse}(T₁,a)$
              } \Else{
                $\mathrm{traverse}(T₁,a); \mathrm{traverse}(T₀,a)$
              }
            }
          }
          \caption{Given a tree $T = T_X$ and $a\in X$, the algorithm yields the elements of $a \xor X = \{a ⊕ x\mid x \in X\}$ in sorted order.}
          \label{algo:traverse}
        \end{algorithm}
    \end{minipage}
    ~
    \begin{minipage}{0.505\textwidth}
      \begin{algorithm}[H]
        \algo{\normalfont{makeTree}($X$)}{
          sort $X$ as $x₁ < \dots < x_n$\;
          let $ℓ_i = x_i ⊕ x_{i+1}, \quad 1 ≤ i < n$\;
          $\mathrm{stream} ← (\infty, x_1, \ell_1, \dots, \ell_{n-1}, x_{n}, \infty)$\;%
          \Return $\mathrm{build}()$ where\;
          \subroutine{\normalfont{build}()}{
            $\ell\gets \mathrm{pop}(\mathrm{stream})$\;
            $x\gets \mathrm{pop}(\mathrm{stream})$\;
            $T ← \mathsf{LeafNode}(x)$\;
            \While{$\mathrm{top}(\mathrm{stream})<\ell$}{
              $\ell'\gets \mathrm{top}(\mathrm{stream})$\;
              $T ← \mathsf{InnerNode}(T,\ell',\mathrm{build}())$\;
            }
            \Return $T$\;
          }
        }
        \caption{\mbox{$\bigO(n\log n)$-time algorithm\quad} \mbox{to construct $T_X$ from $X$.}}
        \label{algo:makeTree}
      \end{algorithm}
    \end{minipage}
\end{figure}

With the data structure $T_X$ in place, the strategy from \threeSUM carries over to \threeXOR as seen in \cref{algo:3xorSimple}.
Summing up, we have obtained the following result: 

\begin{theorem}
  \label{thm:quadraticDeterministic}
  On a deterministic \wordRAM{} the \emph{\threeXOR{}} problem can be solved 
  in time $O(n^{2})$.
  \qed
\end{theorem}
In \cref{algo:makeTree} we provide a linear time construction of $T_X$ from a stream containing the sorted array $X$ interleaved with the labels $ℓ_i = x_i ⊕ x_{i+1}$ (due to sorting the total runtime is $O(n\log n)$). Despite its brevity, the recursive build function is somewhat subtle.

\begin{claim}[Correctness of \cref{algo:makeTree}]
  \label{claim:correctness-make-tree}
  If build() is called while the stream contains the elements $(ℓ_i,x_{i+1},…,x_n,ℓ_n = ∞)$, the call consumes a prefix of the stream until $\mathsf{top}(\mathrm{stream}) = ℓ_j$ where $j = \min\{j > i \mid ℓ_j ≥ ℓ_i\}$. It returns $T_X$ where $X = \{x_{i+1},…,x_j\}$.
\end{claim}
Once this is established, the correctness of makeTree immediately follows as for the outer call we have $i = 0$ and $j = n$ (with the understanding that $∞ ≥ ∞$).
\begin{proof}[Proof of \cref{claim:correctness-make-tree}]
By the $ℓ$-call we mean the (recursive) call to build() with $\mathsf{top}(\mathrm{stream}) = ℓ$. In particular the $ℓ$-call consumes $ℓ$ from the stream and our claim concerns the $ℓ_i$-call. It is clear from the algorithm that an $ℓ$-call can only invoke an $ℓ'$-call if $ℓ' < ℓ$. Therefore the $ℓ_i$-call cannot directly or indirectly cause the $ℓ_j$-call since $ℓ_j ≥ ℓ_i$. At the same time, the $ℓ_i$-call can only terminate when $\mathsf{top}(stream) ≥ ℓ_i$. This establishes that $ℓ_j = \mathsf{top}(stream)$ when the $ℓ_i$-call ends – the first part of our claim.

Next, note that since $X$ is sorted, there is some $m$ such that we have $X₀ = \{x_{i+1},…,x_m\}$ and $X₁ = \{x_{m+1},…,x_j\}$ where $X = X₀ ∪ X₁$ is the partition from the definition of $T_X$. Moreover, $ℓ_m$ is the largest label among $ℓ_{i+1},…,ℓ_{j-1}$. This implies that the $ℓ_m$-call is \emph{directly} invoked from the $ℓ_i$-call. Just before the $ℓ_m$-call is made, the $ℓ_i$-call played out just as though the stream had been $(ℓ_i,x_{i+1},…,x_m,ℓ_m' = ∞)$, which would have produced $T_{X₀}$ by induction\footnote{Formally the induction is on the value of $j-i$. The case of $j-i = 1$ is trivial.}. However, due to $ℓ_m = \mathsf{top}(\mathrm{stream}) < ℓ = ℓ_i$, instead of returning $T = T_{X₀}$, the while loop is entered (again) and produces $\mathsf{InnerNode}(T = T_{X₀}, ℓ = ℓ_m, \mathrm{build}())$. The stream for the $ℓ_m$-call is $(ℓ_m,…,x_n,ℓ_n)$ and $ℓ_j$ is the first label not smaller than $ℓ_m$. So, again by induction, the $ℓ_m$-call produces $T_{X₁}$ and ends with $\mathsf{top}(\mathrm{stream}) = ℓ_j$. Given this, it is clear that afterwards the loop condition in the $ℓ_i$-call is not  satisfied (since $ℓ_j ≥ ℓ_i$) and the new $T = T_X$ is returned immediately, establishing the second part of the claim.
\end{proof}

\section{A Subquadratic Randomized Algorithm}
\label{sec:subquadraticRandomized}

In this section we present a subquadratic expected time algorithm for the \threeXOR{} problem.
Its basic structure is the same as in the corresponding algorithm for \intthreeSUM{} presented in~\cite{BDP08}, 
in particular, it uses buckets and fingerprints, word packing, word-level parallelism, and lookup tables.
Changes are made where necessary to deal with the different setting.
This makes it a little more difficult in some parts of the algorithm 
(mainly because ${\XOR}$-ing a sorted sequence with some $a$ will destroy the order)
and easier in other parts (in particular where linearity of hash functions is concerned).
Altogether, we get an expected running time that is the same as in \cite{BDP08}
for $w=\bigO(\log^2 n)$ and slightly worse for larger $w$.
Recall we assume $w=2^{\bigO(n\log n)}$ throughout.

\begin{restatable}{theorem}{SubquadraticRandomized}
  \label{thm:subquadraticRandomizedShortWords}
  A randomized \wordRAM{} with word length $w$ can solve the \emph{\threeXOR{}} problem
  in expected time 
  \[\bigO\left(n^{2}\cdot\min\left\{\frac{\log^{3}{w}}{w}, \frac{(\log{\log{n}})^2}{\log^2{n}}\right\}\right) \quad \text{for $w=\bigO(n\log n)$,}\]
  and $\bigO(n\log^2 n)$, otherwise.
\end{restatable}
The crossover point between the $w$ and the $\log n$ factor is $w = (\log^2 n)\log\log n$.
The only difference to the running time of~\cite{BDP08} is in an 
extra factor $\log w$ in the word-length-dependent part.

\begin{proof}
  We briefly describe the main ideas of the algorithm.
  For full details, see \cref{sec:app:subquadraticRandomized}.
  If $w=\omega(n\log n)$, we proceed as for $w=\Theta(n\log n)$.
  We use two levels of hashing.

  \subparagraph{Good and Bad Buckets}
  We split $X$ into $R=2^r=o(n)$ \emph{buckets} $\bucket{X}{u}$, $u\in\{0, 1\}^r$, using a randomly chosen hash function $h_1\in\hashFamily^{\text{lin}}_{w,r}$.
  By linearity, for every solution $a\xor b=c$ we also have $h_1(a)\xor h_1(b)=h_1(c)$.
  Given $a\in\bucket{X}{u}$ and $b\in\bucket{X}{v}$, we only have to inspect bucket $\bucket{X}{u\xor v}$ when looking for a $c\in X$ such that $a\xor b=c$.
  
  For $a\in X$, the expected size of bucket $\bucket{X}{h_1(a)}$ is $n/R$.
  A bucket of size larger than $3n/R$ is called \emph{bad}, as are elements of bad buckets.
  All other buckets and elements are called \emph{good}.
  By \cref{cor:hashingLemma}, the expected number of bad elements is smaller than $R$.
  We can even assume that the total number of bad elements is smaller than $2R$.
  (By Markov's inequality, we simply have to repeat the choice of $h_1$ expected $\bigO(1)$ times until this condition is satisfied.)
  
  \subparagraph{Fingerprints and Word-Packed Arrays}
  Furthermore, we use another hash function $h_2\in\hashFamily^{\text{lin}}_{w,p}$ for some appropriately chosen $p$ to calculate $p$-bit \emph{fingerprints} for all elements in $X$.
  If $(3n/R)\cdot p\leq w$, we can pack all fingerprints of elements of a good bucket $\bucket{X}{u}$ into one word $\wpa{X}{u}$.
  This packed representation is called \emph{word-packed array}.
  Again by linearity, for every solution $a\xor b=c$ we have $h_2(a)\xor h_2(b)=h_2(c)$.
  On the other hand, the expected number of \emph{colliding triples}, i.\,e., triples with $a\xor b\neq c$ but $h_1(a)\xor h_1(b)=h_1(c)$ and $h_2(a)\xor h_2(b)=h_2(c)$, is at most $2n^3/(R\cdot 2^p)$.
  
  \medskip
  \noindent The total time for all the hashing steps described so far is $\bigO(n\cdot(r+p))$, see \cref{subsec:LinearUniversalHashing}.
  We consider two choices of $R=2^r$ and $p$, cf.~\cite[proof of Lemma 3]{BDP08} and \cite[proof of Thm. 2]{BDP08}.
  The first one is better for larger words of length $w=\Omega((\log^2{n})\log\log n)$ whereas the second one yields better results for smaller words.
  In both cases, we search for triples with a fixed number of bad elements separately.
  The strategies for finding triples of good elements correspond to the approach for \intthreeSUM{} in~\cite{BDP08}.
  However, for triples with at least one bad element we have to rely on a more fine-grained examination than in~\cite{BDP08}.
  For this, we will use hash tables and another lookup table.
  
  \subparagraph{Long Words: Exploiting Word-Level Parallelism}
  For word lengths $w=\Omega((\log^2 n)\log\log n)$, we choose $R = \ceil{6\cdot n\cdot(\log{w})/w}$ and $p = \floor{2\cdot\log{w}}$ to be able to pack all fingerprints of elements of a good bucket into one word.
  We examine triples with at most one and at least two bad elements separately, as seen in \cref{algo:subquadratic3XOR} in \cref{subsec:pseudocode}.
  
  When looking for triples with at most one bad element, we do the following for every (good or bad) $a\in X$ and $u\in\{0, 1\}^r$ where $\bucket{X}{u}$ and the corresponding bucket $\bucket{X}{h_1(a)\xor u}$ are good (as in \cite[proof of Lemma 3]{BDP08} for all good elements):
  We $\XOR$ every fingerprint of the word-packed array $\wpa{X}{u}$ with $h_2(a)$.
  Then, we apply \cref{lem:subproblemsWordPackedArrays} to get a list of common pairs in this modified word-packed array and $\wpa{X}{h_1(a)\xor u}$.
  For each such pair, we only have to check whether it derives from a non-colliding triple.
  Since we can stop when we find a non-colliding triple
  and since the expected total number of colliding triples is
  $\bigO(n^{2}/(w\log{w}))$,
  we are done in expected time $\bigO(n\cdot R\cdot\log^{2}{w} + n^{2}/(w\log{w})) = \bigO(n^{2}(\log^{3}{w})/w)$.
  (The corresponding strategy in~\cite{BDP08} is only used to examine triples of good elements.)
  
  In order to examine all triples with at least two bad elements,
  we provide a hash table for $X$ with expected 
  construction time $\bigO(n)$ and constant lookup time~\cite{FKS84}.
  Now, for each of the at most $4 R^2 = \bigO(n^2(\log^2 w)/w^2)$ pairs $(a, b)$ of bad elements we can check if $a\xor b\in X$ in constant time.%
  \footnote{Note that it would not be possible to derive expected time $\bigO(R^2)$ for checking all pairs of bad elements
    if we did not start all over if the number of keys in bad buckets is at least $2R$.}
  
  The total expected running time for this parameter choice is $\bigO(n^2(\log^3{w})/w)$.
  
  \subparagraph{Short Words: Using Lookup Tables}
  For word lengths $w=\bigO((\log^2 n)\log\log n)$, we choose $R = \ceil{55\cdot n\cdot(\log{\log n})/\log n}$ and $p = \floor{6\cdot\log{\log{n}}}$ to pack all fingerprints of elements of a good bucket into $(\frac{1}{3}-\varepsilon)\log{n}$ bits, for some $\varepsilon>0$.
  
  We start by looking for triples with no bad element.
  For this, we consider all $\leq R^2$ triples of corresponding good buckets (as in \cite[proof of Thm. 2]{BDP08}).
  We use a lookup table of size $n^{1-\Omega(1)}$ to check whether such a triple of buckets yields a triple of fingerprints (in the word-packed arrays) with $h_2(a)\xor h_2(b)=h_2(c)$ in constant time.
  If this is the case, we search for a corresponding triple $a\xor b=c$ in the buckets of size $\bigO((\log n)/\log{\log{n}})$.
  Since one table entry can be computed in time $\bigO((({\log{n}})/\log{\log{n}})^{3})$,
  setting up the lookup table takes time $n^{1-\Omega(1)}$.
  Furthermore, the expected $\bigO(n^{2}/((\log{\log{n}})\log^5{n}))$ colliding triples cause additional expected running time $\bigO(n^2/((\log\log n)^4\log^2 n)$.
  Since we can stop when we find a non-colliding triple, the total expected time is $\bigO(R^2) = \bigO(n^2(\log\log n)^2/\log^2 n)$.%
  
  Searching for triples with exactly one bad element can be done in a similar way.
  For each bad element $a\in \bad{X}$ and each good bucket $\bucket{X}{u}$, $u\in\{0, 1\}^r$, we $\XOR$ all fingerprints in the word-packed array $\wpa{X}{u}$ with $h_2(a)$ and use a lookup table to check whether it has some fingerprints in common with the word-packed array $\wpa{X}{h_1(a)\xor u}$ of the corresponding good bucket.
  If this lookup yields a positive result, we check all pairs in the corresponding buckets.
  As before, the expected running time is $\bigO(R^{2})$, including the time due to colliding triples.
  
  Examining all triples with at least two bad elements can be done using a hash table as mentioned above in expected time $\bigO(n+R^2)$.
  
  The total expected running time for this parameter choice is $\bigO(n^2(\log{\log{n}})^2/\log^2{n})$.
\end{proof}


\section{Conditional Lower Bounds from the \threeXOR{} Conjecture}
\label{sec:condLowerBounds}

As already mentioned in \Cref{sec:intro}, 
the best \wordRAM{} algorithm for \intthreeSUM{} currently known \cite{BDP08} 
can solve this problem in expected time 
$\bigO(n^{2}\cdot\min\{\frac{\log^{2}{w}}{w}, \frac{(\log{\log{n}})^{2}}{\log^{2}{n}}\})$ for $w=\bigO(n\log n)$.
The best deterministic algorithm \cite{Ch18} takes time $n^2(\log\log n)^{\bigO(1)}/\log^2 n$.
It is a popular conjecture that every algorithm for \threeSUM{} 
(deterministic or randomized) needs (expected) time $n^{2-o(1)}$. 
Therefore, this conjectured lower bound can be used as a basis for conditional lower bounds for a wide range of other problems \cite{GO95,JV16,KPP16,P10}.

Similarly, it seems natural to conjecture that every algorithm for the related \threeXOR{} problem 
(deterministic or randomized) needs (expected) time $n^{2-o(1)}$.
(In \cref{thm:subquadraticRandomizedShortWords}, 
the upper bound for short word lengths is $n^2\frac{(\log{\log{n}})^2}{\log^2{n}} 
= n^{2-(2\log{\log{n}}-2\log{\log{\log{n}}})/\log{n}}$ where $(2\log{\log{n}}-2\log{\log{\log{n}}})/\log{n} = o(1)$.)
Therefore, it is a valid candidate for reductions to other computational problems \cite{JV16,PS16}.

The general strategy from~\cite{BDP08}, already employed in \Cref{sec:subquadraticRandomized}, 
is quite similar to the methods in~\cite{KPP16}.
Therefore, we are able to reduce \threeXOR{} to \offSetDisjointness{} and \offSetIntersection{}, too.
Hence, the conditional lower bounds for the problems mentioned in \cite{KPP16} (and bounds for dynamic problems from \cite{P10}) also hold with respect to the \threeXOR{} conjecture.
A detailed discussion can be found in \cite{PS16}.
Below, we will outline the general proof strategy.

\subsection{\OffSetDisjointness{} and \OffSetIntersection}
\label{subsec:offSetDisInt}

We reduce \threeXOR{} to the following two problems.

\begin{problem}[\OffSetDisjointness]
  \label{prob:offSetDis}
  \textbf{Input:}
  Finite set $C$,
  finite families $A$ and $B$ of subsets of $C$,
  $q\in\mathbb{N}$ pairs of subsets $(S, S')\in A\times B$.
  
  \noindent\textbf{Task:}
  Find all of the $q$ pairs $(S, S')$
  with $S\cap S'\neq\emptyset$.
\end{problem}

\begin{problem}[\OffSetIntersection]
  \label{prob:offSetInt}
  \textbf{Input:}
  Finite set $C$,
  finite families $A$ and $B$ of subsets of $C$,
  $q\in\mathbb{N}$ pairs of subsets $(S, S')\in A\times B$.
  
  \noindent\textbf{Task:}
  List all elements of the intersections $S\cap S'$
  of the $q$ pairs $(S, S')$.
\end{problem}


\subsection{Reductions from \threeXOR}
\label{subsec:reduction}

By giving an expected time $\leq n^{2-\Omega(1)}$ reduction from \threeXOR{} to \offSetDisjointness{} and \offSetIntersection, we can prove lower bounds for the latter two problems, conditioned on the \threeXOR{} conjecture.
\ignore{As in \cite{JV16}, using universe reduction we can assume that our \threeXOR{} instance consists of bitstrings of length $L=\Theta(\log n)$.\footnote{As long as $\log{w}\leq n^{2-\Omega(1)}$, we can perform universe reduction in time $\leq n^{2-\Omega(1)}$ by \cref{lem:evalTimeHashing}(c).}\todo{universe reduction not necessary as long as $\log{w}\leq n^{2-\Omega(1)}$}}

\begin{theorem}
  \label{thm:reductionSetDisjointness}
  Assume \threeXOR{} requires expected time $\Omega(n^2/f(n))$ for $f(n)=n^{o(1)}$ on a \wordRAM{}.
  Then for $0<\gamma<1$ every algorithm for \offSetDisjointness{} that works on instances with
    $\card{C} = \Theta(n^{2-2\gamma})$,
    $\card{A} = \card{B} = \Theta(n\log{n})$,
    $\card{S} = \bigO(n^{1-\gamma})$ for all $S\in A\cup B$ and
    $q = \Theta(n^{1+\gamma}\log{n})$
  requires expected time $\Omega(n^2/f(n))$.
\end{theorem}

\begin{theorem}
  \label{thm:reductionSetIntersection}
  Assume \threeXOR{} requires expected time $\Omega(n^2/f(n))$ for $f(n)=n^{o(1)}$ on a \wordRAM{}.
  Then for $0\leq\gamma<1$ and $\delta>0$, every algorithm for \offSetIntersection{} which works on instances with
  $\card{C} = \Theta(n^{1+\delta-\gamma})$,
  $\card{A} = \card{B} = \Theta(\sqrt{n^{1+\delta+\gamma}})$,
  $\card{S} = \bigO(n^{1-\gamma})$ for all $S\in A\cup B$,
  $q = \Theta(n^{1+\gamma})$ and
  expected output size $\bigO(n^{2-\delta})$
  requires expected time $\Omega(n^2/f(n))$.
\end{theorem}

\begin{proof}
  \label{prf:reductionSet}
  (For more details, see \cite[ch. 6]{PS16}.)
  Let $X\subseteq\{0, 1\}^{w}$ be the given \threeXOR{} instance.
  As in \cref{sec:subquadraticRandomized}, we use two levels of hashing.
  \Cref{algo:setDisjointness,algo:setIntersection} in \cref{subsec:pseudocode} illustrate the reduction to \offSetDisjointness{} and \offSetIntersection, respectively.
  
  At first, we hash the elements of $X$ with a randomly chosen hash function $h_1\in\hashFamily^{\text{lin}}_{w, r}$ into $R = 2^{r} = \Theta(n^\gamma)$ buckets in time $\bigO(n\log n)$.
  Then, we apply \cref{cor:hashingLemma}: There are expected $\bigO(R) = \bigO(n^{\gamma})$ elements in buckets with more than three times their expected size.
  For each such bad element, we can naively check in time $\bigO(n\log{n})$ whether it is part of a triple $(a, b, c)$ with $a\xor b=c$ or not.
  Since $\gamma<1$, all bad elements can be checked in expected time $\leq n^{2-\Omega(1)}$.
  Therefore, we can assume that every bucket $\bucket{X}{u}$, $u\in\{0, 1\}^r$, has $\leq 3\frac{n}{R} = \bigO(n^{1-\gamma})$ elements.
  
  The second level of hashing uses two independently and randomly chosen hash functions $h_{21}, h_{22}\in\hashFamily^{\text{lin}}_{w, p}$ where $P=2^{2p}=(5n/R)^2=\bigO(n^{2-2\gamma})$ for \offSetDisjointness{} and $P=2^{2p}=n^{1+\delta}/R=\bigO(n^{1+\delta-\gamma})$ for \offSetIntersection.
  (The function $h_2$ with $h_2(x) = h_{21}(x)\circ h_{22}(x)$ is randomly chosen from a linear and 1-universal class $\hashFamily$ of hash functions $\{0, 1\}^w\to\{0, 1\}^{2p}$.)
  The hash values can be calculated in time $\bigO(n\log^2{n})$.
  (The additional $\log{n}$ factor is only necessary for \offSetDisjointness{}, since we need to use $\Theta(\log{n})$ choices of hash functions $h_2$ to get an error probability that is small enough.)
  For each $u\in\{0, 1\}^{r}$ and $v\in\{0, 1\}^{p}$, we create ``shifted'' buckets $\bucket{X}{u, v}^{\uparrow} = \condSet{h_2(x)\xor(v\circ 0^{p})}{x\in\bucket{X}{u}}$ and $\bucket{X}{u, v}^{\downarrow} = \condSet{h_2(x)\xor(0^{p}\circ v)}{x\in\bucket{X}{u}}$.
  One such set can be computed in time $\bigO(n^{1-\gamma})$.
  Therefore, all sets can be computed in time $\bigO(R\sqrt{P}\log{n}\cdot n^{1-\gamma}) = O(n^{2-\gamma}\log{n})$ for \offSetDisjointness{} and $\bigO(R\sqrt{P}\cdot n^{1-\gamma}) = O(n^{(3+\delta-\gamma)/2})$ for \offSetIntersection{}.
  
  We can show that for all $u\in\{0, 1\}^r$ and $c\in X$, if there are $a, b\in X$ such that $a\xor b=c$ and $a\in\bucket{X}{u}$, then $\bucket{X}{u, h_{21}(c)}^{\uparrow}\cap\bucket{X}{u\xor h_1(c), h_{22}(c)}^{\downarrow}\neq\emptyset$.
  Therefore, we create the following \offSetDisjointness{} (\offSetIntersection) instance:
    $C := \{0, 1\}^{2p}$,
    $A := \condSet{\bucket{X}{u, v}^{\uparrow}}{u\in\{0, 1\}^r, v\in\{0, 1\}^p}$,
    $B := \condSet{\bucket{X}{u, v}^{\downarrow}}{u\in\{0, 1\}^r, v\in\{0, 1\}^p}$ and
    $q$ queries $(\bucket{X}{u, h_{21}(c)}^{\uparrow}, \bucket{X}{u\xor h_1(c), h_{22}(c)}^{\downarrow})$ for all $u\in\{0, 1\}^r$ and $c\in X$ in time $\leq n^{2-\Omega(1)}$.
    (These are $R\cdot n=\Theta(n^{1+\gamma})$ queries for \offSetIntersection{}. For \offSetDisjointness{}, we create $R\cdot n$ queries for each of the $\Theta(\log{n})$ choices of $h_2$.)
  
  After the \offSetDisjointness{} or \offSetIntersection{} instance has been solved,
  we can use this answer to compute the answer for $X$ in expected time $\leq n^{2-\Omega(1)}$.
  We only have to check if a positive answer from \offSetDisjointness{} (a pair with non-empty intersection) or \offSetIntersection{} (an element of an intersection) yields a solution triple of $X$ or not.
  
  For \offSetDisjointness, we can show that the probability for a triple to yield a false positive can be made polynomially small if we consider $K=\Theta(\log{n})$ choices of $h_2$ and only examine $(X_{u}\xor c)\cap X_{h_1(c)\xor u}$ if this is suggested by all $K$ corresponding queries.
  For \offSetIntersection, the expected number of colliding triples is $\bigO(n^{2-\delta})$.
  By trying to guess a good triple $\Theta(n\log{n})$ times before creating the \offSetIntersection{} instance we can avoid a problem for the expected running time if a \threeXOR{} instance yields an \offSetIntersection{} instance with output size $\omega(n^{2-\delta})$.
  
  For all relevant values of $\gamma$ and $\delta$, the total running time is $\leq n^{2-\Omega(1)}$ in addition to the time needed to solve the \offSetDisjointness{} or \offSetIntersection{} instance.
\end{proof}


\section{Conclusions and Remarks}
\label{sec:conclusionsRemarks}

We have presented a simple deterministic algorithm with running time $\bigO(n^{2})$.
Its core is a version of the Patricia trie for $X\subseteq\{0, 1\}^{w}$, which makes it possible to traverse the set $a\xor X$ in ascending order for arbitrary $a\in \{0, 1\}^{w}$ in linear time.
Furthermore, our randomized algorithm solves the \threeXOR{} problem in expected time $\bigO(n^{2}\cdot\min\{\frac{\log^{3}{w}}{w}, \frac{(\log{\log{n}})^2}{\log^2{n}}\})$ for $w=\bigO(n\log{n})$, 
and $\bigO(n\log^{2}{n})$ for $n\log n\leq w=\bigO(2^{n\log n})$.
The crossover point between the $w$ and the $\log n$ factor is $w = (\log^2 n)\log\log n$.
The only difference to the running time of~\cite{BDP08} is in an 
extra factor $\log w$ in the word-length-dependent part.
This is due to the necessity to re-sort a word-packed array of size $\bigO(w/\log{w})$ in time $\bigO(\log^{2}{w})$ after we have $\XOR$-ed each of its elements with a (common) element.
Finally, we have reduced \threeXOR{} to \offSetDisjointness{} and \offSetIntersection, establishing conditional lower bounds (as in~\cite{KPP16} conditioned on the \intthreeSUM{} conjecture).

A simple, but important observation, which is used in apparently all deterministic subquadratic time algorithms for \threeSUM{}, is
\emph{Fredman's trick}: \[a+b<c+d \iff a-d<c-b\qquad\text{for all } a, b, c, d\in\mathbb{Z}\,.\]
Unfortunately, such a relation does not exist in our setting,
since there is no linear order $\prec$ on $\{0, 1\}^{w}$ such that $a\xor b\prec c\xor d \iff a\xor d\prec c\xor b$ holds for all $a, b, c, d\in\{0, 1\}^{w}$.
Since all elements are self-inverse,
for $a=b=c=0^{w}$ and any $d\in\{0, 1\}^{w}$, we would get $0^{w}\prec d\iff d\prec 0^{w}$.
Is there another, ``trivial-looking'' trick for \threeXOR{}, that establishes a basic approach to solve \threeXOR{} in deterministic subquadratic time?

Another open question is how the optimal running times for \threeSUM{} and \threeXOR{} are related.
At first sight, the two problems seem to be very similar, but the details make the difference.
The observations mentioned above (especially the problem of re-sorting slightly modified word-packed arrays and the possible absence of a relation like Fredman's trick) hint at a larger gap than expected.
On the other hand, the fact that both problems can be reduced to a wide variety of computational problems in a similar way (\eg{} listing triangles in a graph, \offSetDisjointness{} and \offSetIntersection) increases hope for a more concrete dependance.


\bibliography{references}


\appendix
\section{Appendix}

\subsection{Proof of a Hashing Lemma}
\label{subsec:app:hashingLemma}
We prove \cref{lem:hashingLemma} from~\Cref{subsec:hashingLemma}:

\HashingLemma*

\begin{proof}
  \label{prf:lem:hashingLemma}
	As probability space we use
	$\Omega=\{(h, x, y)  \mid h\in\hashFamily, x, y \in S, x\neq y\}$ with the uniform distribution. 
	Fix $t$ with $2\frac{n}{m} < t\leq n$. For $h \in \hashFamily$ we define
  two sets,
	\begin{align}
	\mathcal{B}_h' &=\{B_{h(x)} \mid  x \in S, \abs{B_{h(x)}} < t\} 
	     \text{ \ \ (the set of ``small'' nonempty $h$-buckets)}\,,\nonumber\\
   S_h' &=\{x\in S\mid \abs{B_{h(x)}} < t\} \text{ \ \ \qquad\quad (the set of keys in these $h$-buckets)}\,,\nonumber\\
	\intertext{and three quantities:}
	   \bar p_h &= \prob{x\in S}{\,\abs{B_{h(x)}}\ge t\,}\text{ \ (so $\abs{S_h'}=(1-\bar p_h)n$)}\,,\nonumber\\
     q_h &= \prob{x, y \in S, x\neq y}{h(x)=h(y)}\,,\label{eq:990}\\
		 q_h' &= \prob{x, y \in S_h', x\neq y}{h(x)=h(y)} = 
		\frac{1}{\abs{S_h'}(\abs{S_h'}-1)}\sum_{B\in\mathcal{B}_h'}\abs{B}(\abs{B}-1).\label{eq:1000}
     	\end{align}
	Since the function $z\mapsto z(z-1)$ is convex, the minimum value of the sum
	$\sum_{B\in\mathcal{B}_h'}a_{B}(a_{B}-1)$, 
	taken over all vectors $(a_B)_{B\in\mathcal{B}_h'}$ with nonnegative coefficients $a_B$ that sum to $\abs{S_h'}$,
	is $\abs{\mathcal{B}_h'}\cdot\frac{\abs{S_h'}}{\abs{\mathcal{B}_h'}}\cdot(\frac{\abs{S_h'}}{\abs{\mathcal{B}_h'}}-1)
	= \abs{S_h'}(\abs{S_h'}/\abs{\mathcal{B}_h'}-1)$. 
	Together with $\abs{\mathcal{B}_h'}\le m$ this allows us to conclude from~\eqref{eq:1000} that
	\begin{equation}\label{eq:1010}
	q_h' \ge  \frac{(1-\bar p_h)n/m-1}{\abs{S_h'}-1}\,.
  \end{equation}
	In~\eqref{eq:990}, we split the probability space according to $x\in S_h'$ and $x\notin S_h'$, to obtain: 
	\begin{align*}
	q_h &= \bar p_h \cdot \condProb{x,y \in S, x\neq y}{h(x)=h(y)}{x\in S - S_h'\,}\\ 
	    &\hspace*{12em} + (1-\bar p_h) \cdot \condProb{x,y \in S, x\neq y}{h(x)=h(y)}{x\in S_h'\,} \\
	    &\stackrel{\eqref{eq:1000}}{\ge} \bar p_h\cdot \frac{t-1}{n-1}  +  (1-\bar p_h)\cdot \frac{\abs{S_h'}-1}{n-1}q_h' \\
			&\stackrel{\eqref{eq:1010}}{\ge} \frac{\bar p_h (t - 2\frac{n}{m}) + \frac{n}{m} - 1}{n-1}\,. 
	\end{align*}
	Taking expectations and using 1-universality yields
	\[\frac1m \ge \prob{(h,x,y)\in\Omega}{h(x)=h(y)} =
	\E{h\in \hashFamily}{q_h} > \frac{\E{h\in\hashFamily}{\bar p_h} (t - 2\frac{n}{m}) + \frac{n}{m} - 1}{n}.\]
	Rearranging terms, we get
	\[
	\E{h\in\hashFamily}{\abs{S-S_h'}\,} = \E{h\in \hashFamily}{\bar p_h n} < \frac{n}{t - 2\frac{n}{m}}\,,
	\]
	which is the claimed inequality.
\end{proof}

\subsection{Set Intersection on Unsorted Word-Packed Arrays}
\label{subsec:subproblemsWordPackedArrays}

We prove \cref{lem:subproblemsWordPackedArrays} from~\Cref{subsec:setIntersectionWPA}:

\SetIntersectionWPA*

First, we describe word-parallel sorting. The basic approach is Batcher's bitonic sort. We follow~\cite{AH97}.%
{
  For simplicity of description, assume $k\ell\le w$ and $\ceil{\log{k}}< \ell$.
  Let $x_{0}, \dots, x_{k-1}$ be $k$ $(\ell-1)$-bit strings.
  The strings are stored in a word in such a way that each string is preceded by one extra bit, the \emph{test bit}.
  %
  %
  For convenience, we may even assume that $k$ is a power of two and that $ck\ell\le w$ for some constant $c\in\mathbb{N}_+$
  (use a constant number of words to simulate one longer word, if necessary).
  Thus, a word has $ck$ fields of $\ell$ bits (for the test bit and one entry). The given strings occupy the $k$ rightmost fields.
  Fields $k$, \dots, $ck-1$ serve as \emph{temporary storage}.
  \newlength{\testBitLength}
  \setlength{\testBitLength}{1em}
  \newlength{\numberLength}
  \setlength{\numberLength}{3em}
  \newlength{\numberAndTestBitLength}
  \setlength{\numberAndTestBitLength}{\testBitLength}
  \addtolength{\numberAndTestBitLength}{\numberLength}
  \newlength{\subwordLength}
  \setlength{\subwordLength}{\numberAndTestBitLength}
  \addtolength{\subwordLength}{0.5\subwordLength}
  \newcommand{\fieldFS}{\footnotesize}
  \setlength{\tabcolsep}{0pt}
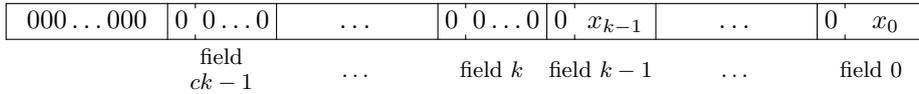
  \begin{center}
    \begin{tabular}{|C{\subwordLength}|C{\testBitLength}:C{\numberLength}|C{\subwordLength}|C{\testBitLength}:C{\numberLength}|C{\testBitLength}:C{\numberLength}|C{\subwordLength}|C{\testBitLength}:C{\numberLength}|}
      \hline
      $000\dots 000$ & $0$ & $0\dots 0$ & $\dots$ & $0$ & $0\dots 0$ & $0$ & $x_{k-1}$ & $\dots$ & $0$ & $x_{0}$ \\
      \hline
      \multicolumn{1}{C{\subwordLength}}{\fieldFS } & 
      \multicolumn{2}{C{\numberAndTestBitLength}}{\fieldFS field $ck-1$} &
      \multicolumn{1}{C{\subwordLength}}{\fieldFS \dots} &
      \multicolumn{2}{C{\numberAndTestBitLength}}{\fieldFS field $k$} &
      \multicolumn{2}{C{\numberAndTestBitLength}}{\fieldFS field $k-1$} &
      \multicolumn{1}{C{\subwordLength}}{\fieldFS \dots} &
      \multicolumn{2}{C{\numberAndTestBitLength}}{\fieldFS field $0$} \\
    \end{tabular}
  \end{center}
}

Let us assume we have packed $k$ numbers $a_{0}, \dots, a_{k-1}\in\{0, 1\}^{\ell-1}$ 
into one word-packed array $a$.
We want to simulate Batcher's \emph{bitonic sort} sorting network to sort these numbers in time $\bigO(\log^{2}{k})$.
If $1\leq g\leq k$ is a power of $2$, we can split $a$ into $\frac{k}{g}$ groups of size $g$ each.
Using the techniques of~\cite[sec. 3]{AH97} (including the use of some constants, which depend on $w$, $k$, and $\ell$
and which can be constructed in time $\bigO(\log w)$), we can solve the following problems:
\begin{itemize}
  \item We can reverse the order of the elements in every group
  in time $\bigO(\log{g})$.
  \item If $g<k$, there is an even number of groups, and
  we can reverse the order of the elements in every second group (with odd (or even) index)
  in time $\bigO(\log{g})$.
  \item If $g>1$, and each group is bitonic,
  we can rearrange the elements in each group in such a way,
  that all the first $g/2$ elements are smaller 
  than all the second $g/2$ elements and both the first and the second $g/2$ elements 
  form a bitonic group of size $g/2$, in \red{time $\bigO(1)$}.
  \item If each group is bitonic,
  we can rearrange the elements in each group so
  that the resulting groups are sorted in increasing order
  in time $\bigO(\log{g})$.
  \item If $g<k$, and each group is sorted ascendingly,
  we can merge the elements of two neighbouring groups (groups $i$ and $i+1$ for $0\leq i\leq\frac{k}{g}-2$)
  in \red{time $\bigO(\log{g})$}.
  \item We can sort the elements in $a$ in increasing order
  in \red{time $\bigO(\log^{2}{k})$}.
\end{itemize}
The sorted word-packed array has its smallest element in field $0$ and its largest element in field $k-1$.

Now, we can check whether two word-packed arrays have a common element in \red{time $\bigO(\log^{2}{k})$}.
Let us assume we have packed a set $A$ of $k$ strings from $\{0, 1\}^{\ell-1}$ 
into the rightmost $k$ fields of one word-packed array $a$ 
and a set $B$ of $k$ strings from $\{0, 1\}^{\ell-1}$ 
into one word-packed array $b$. (There may be some \emph{dummy elements}, \ie, duplicates of elements in $A$ resp. $B$, to reach size $k$.)
We assume $w\geq 4k\ell$. 

With each element, we associate a special \emph{marker bit}, set to $0$ for each element $a\in A$,
and to $1$ for each element $b\in B$.
The marker bit pair is located in the corresponding temporary storage.
We concatenate the two word-packed arrays, resulting in one word $c$ with $2k$ fields and marker bits, which is then
sorted in time $\bigO(\log^{2}{k})$.
(Whenever two fields are swapped, the corresponding fields containing the marker bits are swapped, too.)

It remains to check whether two consecutive fields contain the same value and the 
corresponding marker bits are $0$ and $1$.
For this, we shift $c$ by $\ell$ bits to the right, 
followed by a bitwise $\XOR$ operation with $c$ itself, to get a new word-packed array $c'$.
Then the following statements are equivalent:
(a) $A\cap B\neq\emptyset$,
(b) there are two consecutive elements in $c$ with the same value and marker bit pairs $0$ and $1$, and
(c) $0^{\ell-1}$ is an element of $c'$ with marker bit $1$.
For the two final steps, we sort $c'$ in time $\bigO(\log^{2}{k})$,
treating the marker bit of an element as its least significant bit.
After that, we perform a binary search in time $\bigO(\log{k})$ to check 
whether there are some elements $0^{\ell-1}\circ 1$, \ie{} with value $0^{\ell-1}$ and marker bit $1$.

We can even list $t$ corresponding pairs of elements $(a, b) \in A\times B$ (or their indices) in time $\bigO(\log^{2}{k}+t)$:
For this purpose, for each element in $c$, we additionally attach its corresponding index in $a$ (or $b$) to it 
(in the temporary storage corresponding to its field; we need $\ceil{\log{k}}$ bits per entry) 
as a unique identifier.
The word-packed array $c'$ is modified in the same way.
If we carry this information along through the steps above, 
especially during sorting, we are able to identify all pairs of equal elements 
(of $a$ and $b$).%
\footnote{For each element $0^{\ell-1}\circ 1$ in $c'$ we get one pair of elements in $c$ at positions $i$ and $i+1$ 
  (and the corresponding positions in $a$ and $b$ can be identified in the same way). Due to potential collisions, 
  we have to check if $c$ contains more copies of this common element, and therefore if there are 
  more pairs of elements in $a$ and $b$ with this value. Since $c$ is sorted, 
  these elements have to be directly before position $i$ (for elements from $a$) 
  and directly after position $i+1$ (for elements from $b$).}

\section{A Subquadratic Randomized Algorithm}
\label{sec:app:subquadraticRandomized}

We give a more detailed proof of \cref{thm:subquadraticRandomizedShortWords} from~\Cref{sec:subquadraticRandomized}:

\SubquadraticRandomized*

As mentioned before, for $w=\omega(n\log n)$, we proceed as for $w=\Theta(n\log n)$.

  \subsection{Buckets and Fingerprints}
  \label{subsec:BucketsFingerprints}
  
  We begin by sorting the sets $X\subseteq \{0,1\}^w$ 
  into ascending lexicographic order in time $\bigO(n\log n)$.
  
  Let $R=2^r$ for some $r$.
  For convenience, we identify the sets $[R]$ (integers) and $\{0,1\}^r$ (strings).
  (The value of $r$ will be specified later; we will have $R=o(n)$, hence $r < \log n$.) 
  Now, we choose a hash function 
  $h_1\colon U \to [R]$ from $\hashFamily_{1} = \hashFamily^{\text{lin}}_{w,r}$
  (see~\Cref{subsec:LinearUniversalHashing}). 	
  Function $h_1$ is applied to the elements of $X$. 
  This splits the set into $R$ buckets.
  We write $\bucket{X}{u} = \condSet{x\in X}{h_1(x)=u}$, for $u\in[R]$. 
  The hash values are calculated once and for all and stored for further use.
  Calculating the hash values and the buckets takes time 
  $\bigO(nr)=\bigO(n\log n)$, by~\Cref{lem:evalTimeHashing}(c), 
  using that $r < \log n \le w$. 
  For $a\in X$, the expected size of bucket $\bucket{X}{h_1(a)}$ is $n/R$.
  Since $X$ was sorted, we can assume 
  that each bucket is sorted as well. 
  
  Let $\bad{X}\subseteq X$ be all elements of $X$ in \emph{bad} buckets, 
  \ie, buckets of size larger than $3\frac{n}{R}$, and let $\good{X}=X\setminus\bad{X}$ 
  be all elements in \emph{good} buckets, \ie, buckets of size at most $3\frac{n}{R}$.
  Clearly $\card{\good{X}} \leq n$.
  By \cref{cor:hashingLemma}, we have $\E{h_1}{\card{\bad{X}}} < R$, and 
  by Markov's inequality $\prob{}{\card{\bad{X}} \geq 2R}<\frac12$.
  In the algorithm we check whether $\card{\bad{X}} < 2R$ occurs.
  If not, we start all over by choosing a new hash function $h_1$.
  This maneuver increases the expected running time by at most a constant factor. 
  From here on we can assume that $\card{\bad{X}}<2R$.
  
  Let $a\in X$ and $b\in\bucket{X}{u}$ for $u=h_1(b)$.
  If there is an element $c\in X$ such that $a\xor b=c$, 
  then linearity of $h_1$ implies $h_{1}(c) = h_{1}(a\xor b) = h_{1}(a) \xor h_{1}(b) = h_{1}(a) \xor u$,
  or $a\xor b\in\bucket{X}{h_{1}(a)\xor u}$.
  
  As in~\cite{BDP08}, a second level of hashing inside each bucket is used
  to replace elements by shorter \emph{fingerprints}. 
  If these are short enough, we can pack all fingerprints 
  from a (good) bucket with at most $3{n}/{R}$ elements into one word
  while ensuring a small error probability, \ie{}, a small expected number 
  of \emph{colliding triples} $(a, b, c)\in X^3$ with $a\xor b \neq c$, 
  but $h_{1}(a\xor b)=h_{1}(c)$ and $h_{2}(a\xor b)=h_{2}(c)$.
  
  Let $p$ be the bitlength of the fingerprints and $P=2^{p}$.
  We intend to pack up to $3\frac{n}{R}$ elements into one $w$-bit word, including some additional space,
  so we choose $p=\bigO(w\cdot\frac{R}{n})$. (The constant will be determined below.)
  We pick a hash function $h_{2}$ from $\hashFamily_{2} = \hashFamily^{\text{lin}}_{w,p}$ uniformly at random in time $\bigO(1)$,
  hash all elements in all buckets, which takes time $\bigO(n\cdot p)$, by \cref{lem:evalTimeHashing}(c).
  The total time for all the hashing steps described so far is $\bigO(n\cdot(r+p))$.
  
  Next, we bound the expected number of colliding triples.
  Let $(a, b, c)\in X^3$ with $a\xor b\neq c$.
  Then
  \[\prob{h_{1}, h_{2}}{h_{1}(a\xor b)=h_{1}(c)\wedge h_{2}(a\xor b)=h_{2}(c)} 
  \leq \frac{1}{R}\cdot\frac{1}{P} = \frac{1}{R\cdot 2^{p}}\,,\]
  since $\hashFamily_{1}$ and $\hashFamily_{2}$ are $1$-universal.
  Hence, the expected number of colliding triples is
  \[\sum_{\substack{a, b, c\in X\\a\xor b\neq c}}{\prob{h_{1}, h_{2}}{(a, b, c)\text{ collides}}} 
  \leq \sum_{\substack{a, b, c\in X\\a\xor b\neq c}}{\frac{1}{R\cdot 2^{p}}} \leq \frac{n^{3}}{R\cdot 2^{p}}\,.\]
  Since $\prob{h_{1}, h_{2}}{\card{\bad{X}} < 2R}>\frac{1}{2}$, the expected number of colliding triples conditioned on $\card{\bad{X}} < 2R$ is
  not larger than $2n^3/(R\cdot 2^{p})$.
  
  We consider two choices for $R$ and $p$, cf.~\cite[proof of Lemma 3]{BDP08} and \cite[proof of Thm. 2]{BDP08}.
  The first one is better for larger words of length $w=\Omega((\log^2{n})\log\log n)$ whereas the second one gives us better results for smaller words.
  In both cases, we search for triples with a fixed number of bad elements separately.
  The strategies for finding triples of good elements correspond to the approach for \intthreeSUM{} in~\cite{BDP08}.
  However, for triples with at least one bad element we have to rely on a more fine-grained examination than in~\cite{BDP08}.
  For this, we will use hash tables and another lookup table.
  
  \subsection{Long Words: Exploiting Word-Level Parallelism}
  \label{subsec:pclong}
  
  For word lengths $w=\Omega((\log^2 n)\log\log n)$, we choose $R = \ceil{6\cdot n\cdot(\log{w})/w}$ and $p = \floor{2\cdot\log{w}}$.
  Evaluating the two hash functions for all keys is done in expected time $\bigO(n(\log R+p)) = \bigO(n\log n)$.
  Then, we have $\bigO(n(\log{w})/w)$ good buckets of size $\bigO(w/\log{w})$ as well as $\bigO(n(\log{w})/w)$ bad elements.
  We are able to pack all fingerprints of elements of a good bucket into one word in time $\bigO(R+n)=\bigO(n)$.
  The packed representation of the fingerprints of a bucket $\bucket{X}{u}$ is called \emph{word-packed array} $\wpa{X}{u}$.
  Furthermore, the expected number of colliding triples (conditioned on $\card{\bad{X}} < 2R$) is bounded by
  $2n^{3} / (R\cdot 2^{p}) = \bigO(n^{2}/(w\log{w}))$.
  
  We examine triples with at most one and at least two bad elements separately, as seen in \cref{algo:subquadratic3XOR} in \cref{subsec:pseudocode}.
  
  \subparagraph{Triples with at most One Bad Element}
  \label{subpar:pclong:01bad}
  
  \Wlog, we examine all triples $(a, b, c)\in X^3$ where $b$ and $c$ are good.
  If $a\xor b = c$ and $h_{1}(b)=u$, 
  then $h_{1}(c) = h_{1}(a)\xor u$ and $h_{2}(c) = h_{2}(a) \xor h_{2}(b)$.
  Thus, fingerprint $h_{2}(c)$ occurs in the word-packed array $\wpa{X}{h_{1}(a)\xor u}$.
  It also occurs in $\wpa{X}{u}\xor (h_2(a), h_2(a), \dotsc, h_2(a))$ (each fingerprint in $\wpa{X}{u}$ has been $\XOR$-ed with $h_{2}(a)$).
  
  Hence,
  we run through all $a\in X$ and all $u\in[R]$.
  If $\bucket{X}{u}$ and the corresponding bucket $\bucket{X}{h_1(a)\xor u}$ are good, we search for elements $b\in \bucket{X}{u}$ and $c\in \bucket{X}{h_1(a)\xor u}$ with $a\xor b=c$.
  For this, we first apply $\xor h_2(a)$ to all fingerprints in $\wpa{X}{u}$.
  (This can be done in constant time if we have precalculated a suitable constant in time $\bigO(\log(n/R))=\bigO(\log w)$.)
  Then, we look for all pairs of equal fingerprints in $\wpa{X}{u}\xor (h_2(a), h_2(a), \dotsc, h_2(a))$ and $\wpa{X}{h_{1}(a)\xor u}$.
  If there are $t$ such pairs, we can list them in time 
  $\bigO(t+\log^{2}{\frac{n}{R}}) = \bigO(t+\log^{2}{\frac{w}{\log{w}}}) 
  = \bigO(t+\log^{2}{w})$, by \cref{lem:subproblemsWordPackedArrays}.
  Then, in time $\bigO(t)$, we check each of these $t$ pairs whether it derives from a non-colliding triple.
  Since we can stop after we found a non-colliding triple
  and since the expected total number of colliding triples is
  $\bigO(n^{2}/(w\log{w}))$,
  we are done in expected time $\bigO(n\cdot R\cdot\log^{2}{w} + n^{2}/(w\log{w})) = \bigO(n^{2}(\log^{3}{w})/w)$.
  
  \subparagraph{Triples with at least Two Bad Elements}
  \label{subpar:pclong:23bad}
  
  \Wlog, we examine all triples $(a, b, c)\in X^3$ where $b$ and $c$ are bad.
  Given $b, c\in\bad{X}$, we have to check if there is some $a\in X$ with $a\xor b=c$.
  For this, we create a hash table for $X$ with expected 
  construction time $\bigO(n)$ and constant lookup time~\cite{FKS84}.
  Since there are less than $4R^2$ pairs $(b, c)$,
  the expected time for this check is $\bigO(n + R^2) = \bigO(n + n^2(\log^2 w)/w^2)$.%
  \footnote{Note that it would not be possible to derive expected time $\bigO(R^2)$ for checking all pairs of bad elements
    if we did not start all over if the number of keys in bad buckets is at least $2R$.}
  
  \subsection{Short Words: Using Lookup Tables}
  \label{subsec:pcshort}
  
  For word lengths $w=\bigO((\log^2 n)\log\log n)$, we choose $R = \ceil{55\cdot n\cdot({\log{\log{n}}})/\log{n}}$ and $p = \floor{6\cdot\log{\log{n}}}$.
  Evaluating the two hash functions for all keys is done in expected time $\bigO(n(\log R+p)) = \bigO(n\log n)$.
  Then, we have $\bigO(n({\log{\log{n}}})/\log{n})$ good buckets of size $\bigO(\log{n}/\log{\log{n}})$ as well as $\bigO(n({\log{\log{n}}})/\log{n})$ bad elements.
  We are able to pack all fingerprints of elements in a good bucket into $\leq\delta\log n$ bits, for some constant $\delta\in(0, \frac13)$ in time $\bigO(R+n)=\bigO(n)$.
  Furthermore, the expected number of colliding triples (conditioned on $\card{\bad{X}} < 2R$) is bounded by
  $2n^{3} / (R\cdot 2^{p}) = \bigO(n^{2}/((\log{\log{n}})\log^5{n}))$.
  
  \subparagraph{Triples with No Bad Element}
  \label{subpar:pcshort:0bad}
  
  To find all triples of good elements,
  we use the lookup table strategy from~\cite{BDP08}.
  We consider all pairs of good buckets $\bucket{X}{u}, \bucket{X}{v}\subseteq \good{X}$, both of size $\le 3n/R$,
  so that our algorithm performs at most $R^{2} = \bigO(n^{2}\frac{(\log{\log{n}})^{2}}{\log^{2}{n}})$ rounds. 
  Given $u$ and $v$, only bucket $\bucket{X}{u\xor v}\subseteq\good{X}$ 
  can possibly contain a good $c$ with $a\xor b=c$ and $(a, b) \in \bucket{X}{u}\times \bucket{X}{v}$.
  Instead of searching for a triple $(a,b,c)$ with $a\xor b=c$ naively, 
  we use a lookup table indexed by three word-packed arrays $\wpa{X}{u}$, $\wpa{X}{v}$, $\wpa{X}{u\xor v}$ as a pre-stage.
  (This table has size $o(n)$ and can be built in time $o(n)$.)
  Only if this lookup yields a positive result, 
  we check in time $\bigO(({\log{n}}/\log{\log{n}})^{3})$ whether there is a non-colliding triple in the corresponding buckets.
  We stop as soon as a non-colliding triple is found. 
  Since the expected number of colliding triples is only $\bigO(n^{2}/((\log{\log{n}})\log^5{n}))$, 
  the overall time for all these checks is negligible in comparison to the claimed time bound.
  
  An entry of the lookup table is indexed by a triple $(\alpha, \beta, \gamma)$ of 
  word-packed arrays, each containing $3n/R$ many $p$-bit strings,
  and indicates (by one bit) if there are elements $\alpha_{i}, \beta_{j}, \gamma_{k}$ in these arrays 
  such that $\alpha_{i}\xor \beta_{j} = \gamma_{k}$.
  The number of entries is $2^{3\delta\log n} = n^{3\delta} = n^{1-\Omega(1)}$.
  One table entry can be computed in time $\bigO((({\log{n}})/\log{\log{n}})^{3})$,
  and so setting up the lookup table takes time $n^{1-\Omega(1)}$.
  
  Thus, the total time bound is $\bigO(R^2)$ in the worst case (for the rounds)
  plus $o(n)$ (for setting up the lookup table)
  plus expected time $\bigO(n^2/((\log\log n)^4\log^2 n)$ (for the extra work caused by colliding triples), altogether 
  $\bigO(R^2) = \bigO\bigl(n^{2}\frac{(\log{\log{n}})^{2}}{\log^{2}{n}}\bigr)$.
  
  \subparagraph{Triples with One Bad Element}
  \label{subpar:pcshort:1bad}
  
  \Wlog, we examine all triples $(a, b, c)\in\bad{X}\times\good{X}\times\good{X}$.
  In this case we employ lookup tables just as before,
  but only for pairs of good buckets.
  We treat each pair $(a,u)\in\bad{X} \times [R]$ separately,
  \ie{}, there are $\card{\bad{X}}\cdot R < 2R^{2} = \bigO\left(n^{2}\frac{(\log{\log{n}})^{2}}{\log^{2}{n}}\right)$ rounds.
  We need to look for non-colliding triples $(a, b, c)\in \{a\}\times \bucket{X}{u}\times \bucket{X}{h_1(a)\xor u}$ 
  with $a\xor b=c$, where $\bucket{X}{u}$ and $\bucket{X}{h_1(a)\xor u}$ are good.
  We use a lookup table to check in constant time whether $\wpa{X}{u}\xor (h_2(a), h_2(a), \dotsc, h_2(a))$ and $\wpa{X}{h_1(a)\xor u}$ contain a common element or not.
  If this lookup yields a positive result, 
  we check in time $\bigO((\log^{2}{n})/(\log{\log{n}})^{2})$ 
  whether there is a non-colliding  triple in the corresponding buckets or not.
  Once we have found such a triple, we stop. 
  The expected total number of colliding triples is 
  $\bigO(n^{2}/((\log{\log{n}})\log^5{n}))$,
  and hence the time spent for checking these is smaller than the claimed bound. 
  
  As before, the time for building the lookup table is $n^{1-\Omega(1)}$.  
  So, the total expected time for
  this case is $\bigO(R^{2}) = \bigO(n^{2}{(\log{\log{n}})^{2}}/{\log^{2}{n}})$.
  
  \subparagraph{Triples with at least Two Bad Elements}
  \label{subpar:pcshort:23bad}
  
  As in \cref{subpar:pclong:23bad}, we can use a hash table to handle this case in expected time $\bigO(n+R^2) = \bigO(n^{2}{(\log{\log{n}})^{2}}/{\log^{2}{n}})$.
  
  \bigskip
  
  Since all combinations of good and bad buckets
  give expected running times $\bigO(n^2(\log^3{w})/w)$ 
  and $\bigO(n^2(\log{\log{n}})^2/\log^2{n})$, respectively, 
  \Cref{thm:subquadraticRandomizedShortWords} is proved. \qed

\newpage
\subsection{Pseudocode}
\label{subsec:pseudocode}

For the convenience of the reader, we append some pseudocode implementations of the randomized subquadratic time algorithm and the reductions to \offSetDisjointness{} and \offSetIntersection.

\begin{algorithm}
  \algo{\normalfont{3XOR}($X$)\tikzmark{3XOR}}{
    \Repeat{$\card{B} < 2R$}{
      \tcp{partition $X$ into buckets using $h₁$:}
      pick linear, $1$-universal $h₁: \{0,1\}^w → \{0,1\}^r$ with $2^r = R ≈ \ceil{6n (\log w)/ w}$\;
      $\bucket{X}{u}\gets \{x\in X\mid h_1(x)=u\}$ for $u\in\{0, 1\}^r$\;
      $B ← \{x \in X \mid \card{\bucket{X}{h(x)}} > 3 \frac{n}{R}\}$ \tcp{bad elements in overfull buckets}
    }
    \tcp{search for solution involving at least two bad elements:}
    \For(\tcp*[h]{$<4R^2$ choices}){$a, b \in B$}{
      \If(\tcp*[h]{$O(1)$ using appropriate hash table for $X$}){$a⊕b \in X$}{
        \Return $(a,b,a⊕b)$\;
      }
    }
    \tcp{search for solution involving at most one bad element:}
    $\bucket{X}{u}\gets \emptyset$ for $u\in\{0, 1\}^r$ with $\card{\bucket{X}{u}}>3\frac{n}{R}$ \tcp{empty the bad buckets}
    pick linear, $1$-universal $h₂: \{0,1\}^w → \{0,1\}^p$ with $p = \floor{2 \log w}$\;
    \For{$u \in \{0,1\}^r$}{
      \tcp{pack fingerprints of elements of $\bucket{X}{u}$ into one word $\wpa{X}{u}$}
      $\wpa{X}{u} ← h₂(\bucket{X}{u}) := \textrm{concatenate}\ \{h₂(x) \mid x \in \bucket{X}{u}\}$
    }
    \For(\tcp*[h]{$n·R$ iterations}){$a \in X$ and $u \in \{0,1\}^r$}{
      $\wpax{X}{u}{a} \gets \wpa{X}{u} ⊕ h₂(a)$ \tcp{$h₂(a)$ added to each fingerprint in $\wpa{X}{u}$}
      \For{$v \in \wpax{X}{u}{a} ∩ \wpa{X}{h_1(a)\xor u}$\tikzmark{b2top}}{
        identify responsible $b,c$, in particular with\tikzmark{b2right}\linebreak $v = h₂(a) ⊕ h₂(b) = h₂(c),\ \ h₁(b) = u$\;
        \If{$a ⊕ b = c$}{
          \Return $(a,b,c)$\;
        }\tikzmark{b2bot}
      }
    }
    \Return \emph{no solution}
    \AddNote{[yshift=0.5em]b2top}{b2bot}{[xshift=1em]b2right}{needs time $\bigO(\log²(n/R))$ plus size of intersection}%
    \begin{tikzpicture}[overlay, remember picture]%
      \node[] at ([xshift=5.5em,yshift=0.2\baselineskip]3XOR) {\tcp{$X\subseteq\{0, 1\}^w$, $\card{X}=n$}};%
    \end{tikzpicture}%
  }
  \caption{Our randomized subquadratic \threeXOR{} algorithm from \cref{sec:subquadraticRandomized} for the case $w = Ω((\log² n)\log\log n)$. For $w = o((\log² n)\log\log n)$ using lookup tables to search for solutions involving at most one bad element yields a faster algorithm.}
  \label{algo:subquadratic3XOR}
\end{algorithm}


\begin{algorithm}
  \reduction{\normalfont{3XOR-to-offlineSetDisjointness}($X, \gamma$)\tikzmark{3XORoffSetDisj}}{
    \tcp{partition $X$ into buckets using $h₁$:}
    pick linear, $1$-universal $h₁: \{0,1\}^w → \{0,1\}^r$ with $2^r = R ≈ \ceil{n^\gamma}$\;
    $\bucket{X}{u}\gets \{x\in X\mid h_1(x)=u\}$ for $u\in\{0, 1\}^r$\;
    $B ← \{x \in X \mid \card{\bucket{X}{h(x)}} > 3 \frac{n}{R}\}$ \tcp{bad elements in overfull buckets}
    
    \For(\tcp*[h]{expected $O(R)$ elements}){$b \in B$}{
      $X^{⊕b} ← \mathrm{sort} \{a ⊕ b \mid a \in X\}$\;
      \If{$∃ c \in X^{⊕b} ∩ X$}{
        \Return $(c\xor b,b,c)$
      }
    }
  
    \tcp{create shifted buckets using $h_{21}^{i}$, $h_{22}^{i}$:}
    pick linear, 1-universal $h_{21}^{i},h_{22}^{i} : \{0,1\}^w → \{0,1\}^p$ with $2^{2p} = P \approx \ceil{(5n/R)²} = O(n^{2-2γ})$ and $0\leq i<\ceil{\log{n}}$\;
    \For{$u \in \{0,1\}^r$, $v \in \{0,1\}^p$ and $0\leq i<\ceil{\log n}$}{
      $X_{u,v}^{↑, i} ← \{(h_{21}^{i}(a)⊕v,h_{22}^{i}(a)) \mid a \in X_u\}$\;
      $X_{u,v}^{↓, i} ← \{(h_{21}^{i}(a),h_{22}^{i}(a)⊕v) \mid a \in X_u\}$\;
    }
    
    \tcp{apply algorithm for \offSetDisjointness:}
    $(A,B,C,Q) ← ((X_{u,v}^{↑,i})_{u,v,i},(X_{u,v}^{↓,i})_{u,v,i}, \{0,1\}^{2p}, \emptyset)$\;
    \For{$c \in X$, $u \in \{0,1\}^r$ and $0\leq i<\ceil{\log n}$}{
      $q \gets (X_{u,h_{21}^{i}(c)}^{↑,i},X_{u⊕h₁(c),h_{22}^{i}(c)}^{↓,i})$, identified by $(c, u, i)$\;
      $Q ← Q ∪ \{q\}$
    }
    $Q' ←$ offlineSetDisjointness($A$,$B$,$C$,$Q$) \tcp{$Q'\subseteq Q$}
    
    \tcp{calculate solution for the \threeXOR{} instance:}
    \For{$c\in X$ and $u\in\{0, 1\}^r$}{
      \If{$(c, u, i)\in Q'$ for all $0\leq i<\ceil{\log n}$}{
        $X_u^{⊕c} ← \mathrm{sort} \{a\xor c \mid a \in X_u\}$\;
        \If{$∃b\in X_u^{⊕c} ∩ X_{h₁(c)⊕u}$}{
          \Return $(b\xor c,b,c)$
        }
      }
    }
    \Return \emph{no solution}
    \begin{tikzpicture}[overlay, remember picture]%
      \node[] at ([xshift=7.5em,yshift=0.2\baselineskip]3XORoffSetDisj) {\tcp{$X\subseteq\{0, 1\}^w$, $\card{X}=n$, $0<\gamma<1$}};%
    \end{tikzpicture}%
  }%
  \caption{Algorithm from \cref{subsec:reduction} reducing \threeXOR{} to \offSetDisjointness, establishing a conditional lower bound on the runtime of \offSetDisjointness.}%
  \label{algo:setDisjointness}
\end{algorithm}

\begin{algorithm}
  \reduction{\normalfont{3XOR-to-offlineSetIntersection}($X, \gamma, \delta$)\tikzmark{3XORoffSetInter}}{
    \tcp{$X\subseteq\{0, 1\}^w$, $\card{X}=n$, $0\leq\gamma<1$, $0<\delta<1+\gamma$}
    \tcp{try to guess a solution}
    \SetKwFor{RepeatTimes}{repeat}{times}{}
    \RepeatTimes{$\ceil{\delta n^{\delta}\ln{n}}$}{
      pick $a, b\in X$ independently at random\;
      \If{$a\xor b\in X$}{
        \Return $(a, b, a\xor b)$
      }
    }
    \tcp{partition $X$ into buckets using $h₁$:}
    pick linear, $1$-universal $h₁: \{0,1\}^w → \{0,1\}^r$ with $2^r = R ≈ \ceil{n^\gamma}$\;
    $\bucket{X}{u}\gets \{x\in X\mid h_1(x)=u\}$ for $u\in\{0, 1\}^r$\;
    $B ← \{x \in X \mid \card{\bucket{X}{h(x)}} > 3 \frac{n}{R}\}$ \tcp{bad elements in overfull buckets}
    
    \For(\tcp*[h]{expected $O(R)$ elements}){$b \in B$}{
      $X^{⊕b} ← \mathrm{sort} \{a ⊕ b \mid a \in X\}$\;
      \If{$∃ c \in X^{⊕b} ∩ X$}{
        \Return $(c\xor b,b,c)$
      }
    }
    
    \tcp{create shifted buckets using $h_{21}$, $h_{22}$:}
    pick linear, 1-universal $h_{21},h_{22} : \{0,1\}^w → \{0,1\}^p$ with $2^{2p} = P \approx \ceil{n^{1+\delta}/R} = O(n^{1+\delta-γ})$\;
    \For{$u \in \{0,1\}^r$ and $v \in \{0,1\}^p$}{
      $X_{u,v}^{↑} ← \{(h_{21}(a)⊕v,h_{22}(a)) \mid a \in X_u\}$\;
      $X_{u,v}^{↓} ← \{(h_{21}(a),h_{22}(a)⊕v) \mid a \in X_u\}$\;
      \tcp{for each $y$ in $X_{u,v}^{\uparrow}$ ($X_{u,v}^{\downarrow}$), also maintain a list of elements $a\in X_u$ that generate $y$}
    }
    
    \tcp{apply algorithm for \offSetIntersection:}
    $(A,B,C,Q) ← ((X_{u,v}^{↑})_{u,v},(X_{u,v}^{↓})_{u,v}, \{0,1\}^{2p}, \emptyset)$\;
    \For{$c \in X$ and $u \in \{0,1\}^r$}{
      $q \gets (X_{u,h_{21}(c)}^{↑},X_{u⊕h₁(c),h_{22}(c)}^{↓})$, identified by $(c, u)$\;
      $Q ← Q ∪ \{q\}$
    }
    $Q' ←$ offlineSetIntersection($A$,$B$,$C$,$Q$) \tcp{$Q'\colon Q\to\{C'\mid C'\subseteq C\}$}
    
    \tcp{calculate solution for the \threeXOR{} instance:}
    \For{$c\in X$ and $u\in\{0, 1\}^r$}{
      \For(\tcp*[h]{common element of two shifted buckets}){$y\in Q'((c, u))$}{
        \For{$(a, b)\in X_{u}\times X_{u\xor h_1(c)}$ generating $y\in X_{u,h_{21}(c)}^{↑}\cap X_{u⊕h₁(c),h_{22}(c)}^{↓}$}{
          \If{$a\xor b=c$}{
            \Return $(a,b,c)$
          }
        }
      }
    }
    \Return \emph{no solution}\;
  }
  \caption{Algorithm from \cref{subsec:reduction} reducing \threeXOR{} to \offSetIntersection, establishing a conditional lower bound on the runtime of \offSetIntersection.}
  \label{algo:setIntersection}
\end{algorithm}

\end{document}